\documentclass[alpha-refs,serif]{wiley-article}

\usepackage{siunitx}

\usepackage{booktabs} 
\usepackage{bm} 
\usepackage{varioref} 

\newtheorem{thm}{Theorem} 
\newtheorem{lem}{Lemma} 
\newtheorem{rem}{Remark} 
\def\cov{{\rm Cov}} 
\newcommand{\bbeta}{\mbox{\boldmath $\beta$}} 

\papertype{Original Article}
\paperfield{Journal Section}

\title{Hypothesis testing for quantitative trait locus effects in both location and scale in genetic backcross studies}

\author[1\authfn{1}]{Guanfu Liu}
\author[2\authfn{2}]{Pengfei Li}
\author[3\authfn{3}]{Yukun Liu}
\author[3\authfn{4}]{Xiaolong Pu}

\affil[1]{School of Statistics and Information, Shanghai University of International Business and Economics}
\affil[2]{Department of Statistics and Actuarial Science, University of Waterloo}
\affil[3]{Key Laboratory of Advanced Theory and Application in Statistics and Data Science - MOE,
School of Statistics, East China Normal University.}
\corraddress{Yukun Liu, Key Laboratory of Advanced Theory and Application in Statistics and Data Science - MOE,
School of Statistics, East China Normal University.}
\corremail{ykliu@sfs.ecnu.edu.cn}

\fundinginfo{Dr. Li's research is supported in part by NSERC Grant RGPIN-2015-06592. Dr. Liu and Dr. Pu's research is supported by grants from the National Natural Science Foundation of China (11801359, 11771144, 11771145, 11971300),  the 111 Project (B14019) and Shanghai National Natural Science Foundation (19ZR1420900).}

\runningauthor{Liu, G. et al.}

\begin{document}

\maketitle

\begin{abstract}
Testing the existence of a quantitative trait locus (QTL) effect is an important task in QTL mapping studies. Most studies concentrate on the case where the phenotype distributions of different QTL groups follow normal distributions with the same unknown variance. In this paper we make a more general assumption that the phenotype distributions come from a location-scale distribution family. We derive the limiting distribution of the LRT for the existence of the QTL effect in both location and scale in genetic backcross studies. We further identify an explicit representation for this limiting distribution. As a complement, we study the limiting distribution of the LRT and its explicit representation for the existence of the QTL effect in the location only. The asymptotic properties of the LRTs under a local alternative are also investigated. Simulation studies are used to evaluate the asymptotic results, and a real-data example is included for illustration.
\end{abstract}

\section{Introduction}

Quantitative trait locus (QTL) mapping is an important tool for
analyzing the genetic factors contributing to the variations
of quantitative traits in humans, plants, and animals.
A starting point of  QTL mapping studies is to test  the existence of a QTL effect.
If the QTL effect exists, we proceed to
identify the locations   and estimate the genetic effect of the QTLs.

A popular method  for testing the existence of a QTL effect
is the interval mapping method developed by \cite{Lander:1989}.
Consider a putative QTL, say {\bf Q},   located between two flanking markers
in a backcross design,   with {\bf M} and {\bf N}
being the left  and right flanking markers,  respectively.
For individuals in the backcross population,  the possible  genotypes  are $MM$ and $Mm$ at {\bf M},
$NN$ and $Nn$ at {\bf N}, and $QQ$ and $Qq$ at {\bf Q}.
Hence, the individuals in the backcross population have four marker genotypes:
$MM/NN$, $MM/Nn$, $Mm/NN$, or $Mm/Nn$.
For each individual, we can observe  the genotypes of the two
flanking markers  and the phenotype,
but we cannot observe the  QTL genotype, which must be inferred from the marker information.
The phenotype data can be divided into four groups according to their marker genotypes.
A testing method based on these data for detecting a QTL
in the interval {\bf M}--{\bf N} is referred to as  an interval mapping method.

In this paper, we  consider the backcross study described above and  assume that no interference or  double
recombination occurs between two marker-QTL intervals.
Let $f_1$ and $f_2$ be the phenotype probability density functions corresponding to
the two genotypes $QQ$ and $Qq$ respectively.
We use  $y_{11},\ldots,y_{1n_1}$, $y_{21},\ldots,y_{2n_2}$, $y_{31},\ldots,y_{3n_3}$, and $y_{41},\ldots,y_{4n_4}$ to denote the
phenotype data  corresponding to the marker genotypes $MM/NN$, $MM/Nn$, $Mm/NN$, and $Mm/Nn$, respectively.
 Denote by $r$ and $r_1$   the recombination frequencies between {\bf M} and {\bf N}, and between {\bf M} and {\bf Q}, respectively.
According to  \cite{Doerge:1997} and \cite{Wu:2007},
\begin{equation}
\label{model}
\begin{split}
&y_{1j}\overset{iid}{\sim} f_1(y), \quad j=1,\dots,n_1,\\
&y_{2j}\overset{iid}{\sim} \theta f_1(y)+(1-\theta) f_2(y), \quad j=1,\dots,n_2,\\
&y_{3j}\overset{iid}{\sim} (1-\theta) f_1(y)+\theta f_2(y), \quad j=1,\dots,n_3,\\
&y_{4j}\overset{iid}{\sim} f_2(y), \quad j=1,\dots,n_4,
\end{split}
\end{equation}
where  $\theta=r_1/r$. Since the two flanking markers are prespecified,
the recombination frequency $r$ is generally known. However,
the location of {\bf Q} is unknown, so neither
the recombination frequency $r_1$ nor   the parameter $\theta$ is known.
The problem of testing the existence of a QTL effect is equivalent to testing $H_0$:
$f_1=f_2$   versus  $H_1$: $f_1\neq f_2$.

The likelihood ratio test (LRT) has been commonly applied to test the existence of a QTL effect
when $f_1$ and $f_2$ in (\ref{model}) are from the same parametric distribution family \citep{Chen:2005, Wu:2007}.
Due to the mixture model structure in (\ref{model})
and the fact that $\theta$ appears under only the alternative model,
determining the critical values  of the LRT has been a long-standing problem
in its application.
Under the assumption that $f_1$ and $f_2$ are normal distributions with the same unknown variance,
\cite{Feingold:1993} and \cite{Rebai:1994} proposed approximation methods,
which do not have rigorous theoretical support,
to determine the critical values of the LRT.
Under the same setup, \cite{Chen:2005} rigorously showed that the LRT statistic converges in distribution
to the supremum of a chi-square process under the null hypothesis.
They further suggested using numerical approximation to determine the critical values of the LRT from its limiting distribution.
Simulation shows that their method provides more accurate  critical values  than those
of \cite{Feingold:1993} and \cite{Rebai:1994}.
Similar limiting distributions for the LRT were obtained
by \cite{Wu:2008} and \cite{Kim:2013} for the case
where
$f_1$ and $f_2$ are from a one-parameter (mean parameter) exponential family.
Under the same assumption on $f_1$ and $f_2$ as in \cite{Chen:2005},
\cite{Chang:2009} developed a score test
and  showed that the maximum of the squared score statistic is asymptotically equivalent to  the LRT statistic.

In summary, the aforementioned works concentrate on the QTL effect in the mean parameter only.
As \cite{Weller:1992} and \cite{Korol:1996} have pointed out,
a QTL effect may be economically more critical in the variance
than in the mean (e.g., for earliness, flowering time, ripening time under machine harvesting,
time for chickens to hatch, and seed dormancy).
To the best of our knowledge, \cite{Korol:1996} were the first
to take the variance effect into account in interval mapping studies:
they investigated the power of the LRT through simulations.
Their study was based on the known-QTL-location and
normality assumptions; no  theoretical results were provided
on how to determine critical values of the LRT.

In this paper, we fill the gap in genetic backcross studies by studying
the {LRT} procedure for the existence of a  QTL effect in both location
and scale with an unknown QTL location.
The normality assumption on $f_1$ and $f_2$ is quite natural, but  it can be restrictive.
We instead assume that $f_1$ and $f_2$ come from a general location-scale distribution family,
and they may have different locations and/or scales.
Testing $H_0: f_1=f_2$ under the above setup in (\ref{model}) is essentially testing homogeneity   in four samples.
There has been much research into the asymptotic properties
of the LRT for homogeneity in the mixture model in the one-sample case;
see \cite{Dacunha:1999},
\cite{Liu:2003}, \cite{Garel:2005}, \cite{Gu:2018}, and the references therein.
However, the mixture model with component densities from  a general location-scale distribution family
has some undesirable properties.
The likelihood function of the unknown parameters is unbounded
(\citealp{Hathaway:1985}), and the Fisher information on the mixing proportion
can be infinity (\citealp{Chen:2009}).
Because of these two nonregularities,
the existing elegant asymptotic results of the LRT for homogeneity in the one-sample mixture model
are not directly applicable to the mixture model with component densities from  a  general location-scale distribution family.
Taking advantage of the specific four-sample structure in (\ref{model}),
we successfully derive the asymptotic properties of the  maximum likelihood estimators (MLEs)
of the unknown parameters and the LRT statistic under the null hypothesis that no QTL exists.
To the best of our knowledge,  we are the first to consider the asymptotic results of the LRT for homogeneity
in  the mixture model with component densities from a  general location-scale distribution family.

We focus on the data structure in (\ref{model}) with
$f_1$ and $f_2$ from a general location-scale distribution family.
Our contributions can be summarized as follows.
First, we establish  the convergence rates  of the MLEs
of the location and scale parameters in $f_1$ and $f_2$,
and we show that the LRT statistic  converges in distribution to  the supremum of a $\chi^2_2$ process,
under the null hypothesis that no QTL exists.
We further identify an explicit representation of  the limiting distribution,
which can be used to rapidly calculate  the critical values of the LRT.
Second,
as a complement, we study the limiting distribution of the LRT and its explicit representation
for the existence of  a QTL effect in the location only
(i.e., $f_1$ and $f_2$ have the same unknown scale parameter).
Third, we derive  the local powers of the above two LRTs
under a series of local alternatives.  The local power results indicate
that  the two LRTs are consistent under the given local alternatives.
We emphasize that the existing  asymptotic results
under the normality assumption with a common unknown variance
rely on the assumption that  the mean parameter space is bounded
(\citealp{Chen:2005};  \citealp{Chang:2009}).
Our asymptotic results do not depend on this assumption  whether or not a QTL effect in the scale exists.
Fourth, based on the limiting distributions of two LRTs, the approximation of \cite{Davies:1987}
can be applied to approximate the critical values of the LRTs.
We conduct simulations to show that the  empirical type I errors of two LRTs based on the explicit representations
are quite close to the nominal levels, and they are closer than those based on Davies's method.
Further, the LRT in both the location and scale is uniformly more powerful than existing nonparametric tests such as
the multiple-sample Kolmogorov--Smirnov test (\citealp{Kiefer:1959}) and the multiple-sample Anderson--Darling test (\citealp{Scholz:1987}).

The paper is organized as follows.
Section 2 presents the large-sample properties of the MLEs
of the unknown parameters  and  the LRTs where
(1) $f_1$ and $f_2$ may have different locations and/or scales,
and (2) $f_1$ and $f_2$ have the same unknown scale.
We  give  explicit representations of  the limiting distributions of the LRTs in these two cases,
and we study their local powers   under a series of local alternatives.
Section 3 investigates the finite-sample performance of the LRTs via simulation studies,
and Section 4 analyzes a real-data set.
Section 5 concludes with a discussion.
For clarity, the proofs are postponed to the Appendix  or the supplementary material.

\section{Main results}

Suppose we have the observations
$\{y_{ij},~i=1,\ldots,4,~j=1,\ldots,n_i\}$
from  (\ref{model})
with $f_1(y)$ and $f_2(y)$ from the
same location-scale distribution family.
That is,
  $f_1(y)=f(y;\mu_1,\sigma_1)$ and $f_2(y)=f(y;\mu_2,\sigma_2)$
  with
  $f(y;\mu,\sigma)= \sigma^{-1}f\big((y-\mu)/\sigma;0,1 \big)$,
 where
 $f(\cdot; 0, 1)$ is a known probability density function,
 and $\mu$ and $\sigma$ are the location and scale parameters, respectively.
Under this setup,
testing the existence of the QTL effect in both location and scale
is equivalent to testing
\begin{equation}
H_0: (\mu_1,\sigma_1)=(\mu_2,\sigma_2).
\label{null1}
\end{equation}

\subsection{Asymptotic properties under the null}
The LRT is one of the most important tools in statistical inference, especially for
parametric models  (\citealp{Wilks:1938}; \citealp{Chernoff:1954}; \citealp{Self:1987}).
In this subsection, we establish the LRT statistics for testing (\ref{null1}).
Based on the observed data in  (\ref{model}),
the log-likelihood function of
$(\theta, \mu_1, \mu_2, \sigma_1, \sigma_2)$   is
\begin{align*}
l_n(\theta, \mu_1, \mu_2, \sigma_1, \sigma_2)
=&  \sum\limits_{j=1}^{n_1}\log f_1(y_{1j})
 +\sum\limits_{j=1}^{n_{2}}\log
 \{ \theta f_1(y_{2j})+(1-\theta)f_2(y_{2j}) \} \nonumber\\
& +\sum\limits_{j=1}^{n_{3}}\log \{ (1-\theta) f_1(y_{3j})
+\theta f_2(y_{3j}) \}
+\sum\limits_{j=1}^{n_{4}}\log f_2(y_{4j}).
\label{log-likelihood}
\end{align*}
The MLEs of the unknown parameters under the null and full models are  respectively
\begin{equation}
\label{MLE-null}
(\hat\mu_0,\hat\sigma_0)=\arg\max_{\mu,\sigma}
l_n(0.5,\mu,\mu,\sigma,\sigma)
\end{equation}
and
$$
(\hat\theta,\hat\mu_1,\hat\mu_2, \hat\sigma_1, \hat\sigma_2)
=
\arg\max_{\theta, \mu_1, \mu_2, \sigma_1, \sigma_2}
l_n(\theta, \mu_1, \mu_2, \sigma_1, \sigma_2).
$$
Then the LRT statistic for testing (\ref{null1}) is defined to be
$$
R_n=2\{l_n(\hat\theta,\hat\mu_1,\hat\mu_2, \hat\sigma_1, \hat\sigma_2)-l_n(0.5,\hat\mu_0,\hat\mu_0,\hat\sigma_0,\hat\sigma_0)\}.
$$
We reject $H_0$ when $R_n$ exceeds some critical value to be determined
by its limiting distribution.
In the definition of the MLEs in (\ref{MLE-null})
under the null model,  we arbitrarily set
the parameter $\theta$ to  0.5; other choices of $\theta$ do not change
the resulting  likelihood or   LRT, since
$\theta$ does not appear in the null model.

We define some  notation to ease the presentation of
the asymptotic properties of $(\hat\mu_1,\hat\mu_2, \hat\sigma_1, \hat\sigma_2)$ and $R_n$.
Let the total sample size be  $n=\sum_{i=1}^4n_i$.
We assume that the $n_i$'s go to $\infty$ at the same rate.
That is, $n_i/n$ goes to a constant $p_i$ with $p_i>0$, $i=1,\ldots,4$.
In the genetic backcross studies described in Section 1, the $p_i$'s are related to $r$, the recombination frequency between two markers {\bf M} and {\bf N}, in the following way (see \citealp{Wu:2007}):
$$(p_1,p_2,p_3,p_4)=\left(\frac{1-r}{2}, \frac{r}{2}, \frac{r}{2}, \frac{1-r}{2} \right).$$
Let $z_{hk}$ ($h, k=1,2$) be independent and identically distributed standard normal random
variables, and  define  for $h=1, 2$,
\begin{equation}
\label{zix}
Z_h(\theta)=\dfrac{\sqrt{1-r}}{ \sqrt{1+ 4r\theta (\theta-1)}} z_{h1}
+\dfrac{\sqrt{r}(2\theta-1)}{\sqrt{1+ 4r\theta(\theta-1)}} z_{h2},
\end{equation}
 where   $0\leq \theta\leq 1$.
It is clear that the $\{Z_h(\theta):  \theta\in [0, 1]\}$ ($h=1, 2$)
are independent and both are Gaussian processes with zero mean,
unit variance, and  covariance function
\begin{equation*}
\cov\big(Z_h(\theta_1), Z_h(\theta_2)\big)
= \dfrac{1 + r\{ 4\theta_1\theta_2-2(\theta_1+ \theta_2)\}}
{ \sqrt{\{ 4r\theta_1(\theta_1-1)+1 \}\{4r\theta_2
(\theta_2-1)+1 \} }}.
\end{equation*}
Let
$\gamma = \arccos\sqrt{1-r}$.  Define three sets of angles,
\begin{align}
&A_1=[-\gamma, \gamma]
\cup[\pi-\gamma,\pi]
\cup[-\pi,-\pi+\gamma],  \nonumber \\
&A_2=[\gamma,\pi/2]
\cup[-\pi+\gamma,-\pi/2], \nonumber \\
&A_3=[\pi/2,\pi-\gamma]
\cup[-\pi/2,-\gamma],
\label{set-A}
\end{align}
which form a partition of $[-\pi,\pi]$
and are depicted  in Figure \ref{figure1}.

\begin{center}
Figure 1 should be inserted here.
\end{center}

In Theorem \ref{lrt-mean-var} we establish the root-$n$ consistency of the MLE of $(\mu_1, \mu_2, \sigma_1, \sigma_2)$
   and   the limiting distribution of $R_n$.
For presentational continuity, we have put the
long proof  in the Appendix and the supplementary material.

\begin{thm}
\label{lrt-mean-var}
Suppose that $f(y;\mu,\sigma)$ satisfies  Conditions A1--A7 in the Appendix and
that  $n_i/n$ goes to $p_i \in (0, 1)$ as $n\to\infty$, $i=1,\ldots,4$.
Under the null distribution $f(y;\mu_0,\sigma_0)$, we have that
\begin{itemize}
\item[(i)] $ \hat\mu_h=\mu_0+O_p(n^{-1/2})$  and
$ \hat\sigma_h=\sigma_0+O_p(n^{-1/2})$, $h=1,2$;
\item[(ii)] as $n \rightarrow \infty$,  the LRT statistic $R_n\stackrel{d}\rightarrow
R=\sup_{0\leq \theta\leq 1}\{ Z^2_1(\theta)+Z^2_2(\theta)\}$,
where  $\stackrel{d}\rightarrow$ stands for convergence in distribution and
 $Z_1(\theta)$ and  $Z_2(\theta)$ are defined in \eqref{zix}.
\end{itemize}
Further, let $\rho_h^2$ with $\rho_h>0$ and $h=1,2$
be independent random variables from $\chi_2^2$.
Define  $\eta=(U_1+U_2)/2-(\pi/4)$, where $U_1$  and $U_2$, independent of $\rho_1^2$ and $\rho_2^2$,
are independent
random variables from the uniform distribution on $[-3\pi /4, 5\pi/4]$.
Then,
\begin{eqnarray}
\label{rep1}
\hspace{-0.1in}
R
\stackrel{d}{=}
\frac{\rho_1^2 +\rho_2^2}{2}   +
 \rho_1 \rho_2
\{ I_{A_1} (\eta)
+  I_{A_2} (\eta) \cos (2\eta-2\gamma)     + I_{A_3} (\eta) \cos (2\eta+2\gamma) \},
\end{eqnarray}
where $X\stackrel{d}{=}Y$ indicates that the
 two random variables $X$ and $Y$ have the
same distribution.
\end{thm}

Developing the asymptotic results of $R_n$ is technically challenging for two reasons.
First, $\theta$ is a nuisance parameter that appears only in the alternative model and
hence is not identifiable under the null model.
This invalidates many elegant asymptotic results for the classical
LRT method (\citealp{Wilks:1938}; \citealp{Chernoff:1954}; \citealp{Self:1987}).
Second, the presence of scale parameters in the model also complicates the derivation;
see the comments in  \cite{Chen:2003} and  \cite{Chen:2009}.
The log-likelihood functions for the second and third groups of observations are  unbounded.
Using the first and fourth groups of observations,
we are able to show in Lemma \ref{lem1} in the Appendix (see the supplementary material for a detailed proof)
that any estimator of $(\theta,\mu_1, \mu_2, \sigma_1, \sigma_2)$
with a large likelihood value is consistent for  $\mu_h$ and $\sigma_h$, $h=1,2$.
This implies the consistency of $\hat\mu_h$ and $\hat\sigma_h$ without the condition
that the parameter space for $\mu_h$ and/or $\sigma_h$ is bounded.

The limiting distribution of $R_n$ involves the supremum of the $\chi^2_2$ process.
It may be difficult to calculate the critical values in general (\citealp{Adler:1990}).
The explicit representation of the limiting distribution in (\ref{rep1}) considerably
eases this burden.
For a large number $N$, we can generate $N$ groups of $(\rho_1^2,\rho_2^2,U_1,U_2)$,
from which we obtain $N$ realizations of $R$: $R^{(1)},\ldots,R^{(N)}$.
The quantiles of $R$ can be well approximated by those of $R^{(1)},\ldots,R^{(N)}$.
This method provides a fast way to obtain the critical values of $R$:
it takes less than one minute when $N=100,000$.
The approximation method of \cite{Davies:1987} can be used to find
the approximate quantiles of the supremum of the $\chi^2_2$ process.
The simulation studies in Section 3.2 show that  our explicit representation
results in an LRT with a more accurate type I error rate than
that from Davies's method.

It is worth pointing out that the regularity conditions on $f(y;\mu,\sigma)$ are not restrictive.
The location-scale distributions generated by
the commonly used normal,  logistic, extreme-value, and $t$ distributions
all satisfy Conditions A1--A7; see  the supplementary material.
Hence, the results in Theorem~\ref{lrt-mean-var} are applicable to situations where  $f(y;\mu,\sigma)$ comes from any of these distributions.

As a comparison, we further consider the asymptotic properties of
the LRT test under the assumption that $\sigma_1=\sigma_2$.
The LRT test statistic in this case is defined as
$$
R_n^*=2\{l_n(\hat\theta^*,\hat\mu_1^*,\hat\mu_2^*, \hat\sigma^*, \hat\sigma^*)-l_n(0.5,\hat\mu_0,\hat\mu_0,\hat\sigma_0,\hat\sigma_0) \},
$$
where
$$
(\hat\theta^*,\hat\mu_1^*,\hat\mu_2^*, \hat\sigma^*)=
\arg\max_{\theta, \mu_1, \mu_2, \sigma}
l_n(\theta, \mu_1, \mu_2, \sigma, \sigma).
$$
We present  the asymptotic properties of $(\hat\theta^*,\hat\mu_1^*,\hat\mu_2^*, \hat\sigma^*)$ and $R_n^*$
in Theorem~\ref{lrt-mean}. Its proof is similar to that of Theorem~\ref{lrt-mean-var}
and is omitted to save  space.

\begin{thm}
\label{lrt-mean}
Assume the conditions of Theorem \ref{lrt-mean-var}.
Under the null distribution $f(y;\mu_0,\sigma_0)$, we have
\begin{itemize}
\item[(i)] $\hat\mu_h^*=\mu_0+O_p(n^{-1/2})$ ($h=1,2$) and
$ \hat\sigma^*=\sigma_0+O_p(n^{-1/2})$;
\item[(ii)] as $n\rightarrow \infty$,
the LRT statistic
$R_n^*\stackrel{d}\rightarrow
R^*=
\sup_{0\leq \theta\leq 1} \{ Z^2_1(\theta) \},
$
 where $Z_1(\theta)$  is defined in  \eqref{zix}.
\end{itemize}
Further, suppose $\rho^2$ and $\eta^*$ are two independent random variables that
follow  $\chi_2^2$ and  the uniform distribution on $[-\pi,\pi]$, respectively. Then
\begin{equation}
\label{rep2}
R^*
\stackrel{d}{=}
\rho^2 \{ I_{A_1} (\eta^*) +  I_{A_2} (\eta^*) \cos^2(\eta^*-\gamma) + I_{A_3} (\eta^*) \cos^2(\eta^*+\gamma) \},
\end{equation}
where the $A_i$'s are defined in \eqref{set-A}.
\end{thm}

Compared with the results in \cite{Chen:2005}, Theorem \ref{lrt-mean}
makes two additional contributions.
First, the results are applicable to the more general location-scale distribution family,
whereas the results of \cite{Chen:2005} are restricted to the normal family.
Second,
from Lemma \ref{lem1} in the Appendix, $(\hat\mu_1^*,\hat\mu_2^*,\hat\sigma^*)$
is consistent without the assumption that the parameter space
for $(\mu_1,\mu_2)$ is bounded.
Hence, the asymptotic result for $R_n^*$ does not depend on this restrictive assumption.
The explicit representation in (\ref{rep2}) can be used as in (\ref{rep1})
to calculate the critical values of $R^*$.
\cite{Zhang:2008} also identified a representation for the $\chi^2_1$ process in Part (ii) of
Theorem \ref{lrt-mean}.
Our representation in (\ref{rep2}), obtained by polar transformations,
is a refined version of theirs.

\subsection{Local power analysis}

Our previous analysis guarantees that  in theory  the LRT with the proposed
critical value determining strategy can control its type I error asymptotically
when the null hypothesis is true.
We may wonder how it performs when  the null hypothesis is violated.
Asymptotic local power analysis  has become an important and
increasingly used tool for this purpose.

To investigate the asymptotic local power of $R_n$ and $R_n^*$, we consider
the following local alternative
\begin{equation}
H_{A}^{n}: \theta=\theta_0, ~
\begin{pmatrix}\mu_1\\ \mu_2\end{pmatrix}=
\begin{pmatrix}\mu_0-n^{-1/2}\delta_{\mu}\\
\mu_0+n^{-1/2}\delta_{\mu}\end{pmatrix}, ~
\begin{pmatrix}\sigma_1\\ \sigma_2\end{pmatrix}=
\begin{pmatrix}\sigma_0-n^{-1/2}\delta_{\sigma}\\
\sigma_0+n^{-1/2}\delta_{\sigma}\end{pmatrix},
\label{alternative}
\end{equation}
where $\delta_{\mu}$ and $\delta_{\sigma}$  are both positive constants.
Define  $\tau(\theta)=1+4r\theta(\theta-1)$,
\begin{equation*}
\label{def.TU}
T=\frac{\partial f(y_{11}; 0, 1)/\partial \mu} {f(y_{11},0,1)},
\quad \mbox{and}\quad
U=\frac{\partial
f(y_{11}; 0, 1)/\partial \sigma} {f(y_{11};0,1)}.
\end{equation*}
Let $\sigma_T^2=E(T^2)$  and let ${\bm A}$ be the covariance matrix of $(T,U)$,
where  the expectation and covariance are taken with respect to $f(y;0,1)$.
Further, let $\chi^2_m(c)$  denote  the non-central  chi-square distribution with non-centrality parameter
$c$ and $m$ degrees of freedom.

\begin{thm}
\label{localpower}
Assume the conditions of Theorem \ref{lrt-mean-var}. Under the local alternative $H_{A}^n$ in (\ref{alternative}), we have
\begin{itemize}
\item[(i)] $
    R_n\stackrel{d}\rightarrow
    \sup_{\theta\in[0,1]}\{\chi_2^2\big(\bm \rho_{\theta_0}^\tau(\theta)\bm\rho_{\theta_0}(\theta)\big)\},
    $
    where $$
    \bm\rho_{\theta_0}(\theta)=
    -\{1+2r(2\theta_0\theta-\theta_0-\theta)\}
    \{\tau(\theta)\}^{-\frac{1}{2}}\sigma_0^{-1}\bm A^{\frac{1}{2}}
    \begin{pmatrix}
    \delta_{\mu}\\
    \delta_{\sigma}
    \end{pmatrix};
    $$
\item[(ii)]
    $
    R_n^*\stackrel{d}\rightarrow\sup_{\theta\in[0,1]}\{\chi_1^2\big( \rho_{\theta_0}^{*2}(\theta)\big)\},
    $
    where
    $$
    \rho^*_{\theta_0}(\theta)
    =-\{1+2r(2\theta_0\theta-\theta_0-\theta)\}
    \{\tau(\theta)\}^{-\frac{1}{2}}\sigma^{-1}_0\sigma_T\delta_{\mu}.
    $$
\end{itemize}
\label{thm3}
\end{thm}

For convenience of presentation, the proof of Theorem \ref{thm3} is deferred to the Appendix.
Theorem \ref{localpower} indicates that
 the two LRTs $R_n$ and $R_n^*$ are both consistent
 under the local alternative $H_A^n$.
Note that  $\delta_\sigma$ appears in the limiting distribution of $R_n$
but not in that of $R_n^*$.
Hence, we expect that $R_n$ is more powerful than $R_n^*$ when $\sigma_1$ and $\sigma_2$ are significantly different, i.e., $\delta_\sigma$ is significantly different from 0.
This is confirmed in the simulation study.

\section{Simulation study}
\subsection{Setup}
In this section, we conduct  Monte Carlo simulations to provide insight
into the following questions:
\begin{itemize}
\item[(a)] Do  the limiting distributions of $R_n$ and $R_n^*$
provide  accurate approximations to their finite-sample distributions?
\item[(b)] How do $R_n$ and $R_n^*$ perform in terms of statistical power for detecting the existence of QTL effects?

\end{itemize}

In our simulation studies, we concentrate on the normal and  logistic kernels.
We consider four total sample sizes,  $n=$
50, 100, 200, and 300,
and three values   5, 10, and 20 for $d$, the
inter-marker distance of the single interval.
We use  the Haldane map function
$r = 0.5(1-e^{-2d/100})$ to determine the recombination frequency $r$.
The sample-size vector
$(n_1, n_2, n_3, n_4)$ is generated from a multinomial
distribution
$Multinom(n; (1-r)/2, r/2, r/2, (1-r)/2)$.
In the simulations, we set the significance level to $\alpha=5\%$ and $ 1\%$.
To save space,  the simulation results for $d=10$ are omitted.
Since detecting the existence of the QTL effects is essentially
testing the homogeneity of the distributions in four samples,
multiple-sample nonparametric tests  can be applied.
We choose the multiple-sample analog of the Kolmogorov--Smirnov
test (\citealp{Kiefer:1959}; denoted $KS$) and the multiple-sample
Anderson--Darling test  (\citealp{Scholz:1987};  denoted $AD$) as competitors.
 We study the finite-sample performance of the two LRTs, $R_n$ and  $R_n^*$,
by comparing them with the  two nonparametric tests.

\subsection{Comparison of Type I errors}
We first check the performance of the limiting distributions.
There are two methods to obtain the quantiles of $R_n$ and  $R_n^*$ from their limiting distributions: the first (denoted ``Ours") is based on the explicit representations
in  (\ref{rep1}) and (\ref{rep2}) and the second (denoted ``Davies") is the approximation method of \cite{Davies:1987} for the supremum of the $\chi^2$ process.
When we apply the first method, we generate $N=100,000$ realizations from the explicit representations in \eqref{rep1} and \eqref{rep2} respectively to determine the critical values of $R_n$ and $R_n^*$.
We summarize the empirical type I errors of $R_n$ and  $R_n^*$ from the above two methods and those of  $KS$ and $AD$ calibrated by their limiting distributions.

For the simulation results, the data are generated from $N(0,1)$ and $Logistic(0,1)$, i.e., $f_1=f_2=N(0,1)$ and $f_1=f_2=Logistic(0,1)$, respectively.
Here $N(\mu,\sigma^2)$ denotes normal distribution with mean $\mu$ and variance $\sigma^2$
and $Logistic(\mu,\sigma)$ denotes logistic distribution with location and scale parameters being $\mu$ and $\sigma$, respectively.
All the empirical type I errors in Table \ref{typeI1} are calculated based on 10,000 repetitions.
We can see that the empirical type I errors of the two LRTs based on the explicit representations
in  (\ref{rep1}) and (\ref{rep2})   are quite close to the nominal levels, and they are closer than those based on Davies's method.
The empirical type I errors of
$R_n$ and $R_n^*$ based on the explicit representations
indicate that  the limiting distributions of $R_n$ and $R_n^*$
provide  accurate approximations to their finite-sample distributions.
The simulated type I errors of $AD$ are also quite close to the nominal level,
while $KS$ seems to be quite conservative.
As the sample size increases, all the empirical
type I errors become closer to the nominal levels.

\begin{center}
Table \ref{typeI1} should be inserted here.   
\end{center}

\subsection{Comparison of powers}

In this subsection, we compare the powers of our LRT methods with those of $KS$ and $AD$.
We consider two values for $\theta$, 0.5 and 0.7,
and six combinations of $f_1$ and $f_2$:
\begin{enumerate}
\item[] Case I: $f_1=N(0,1)$ and $f_2=N(0.5,1)$;
\item[] Case II: $f_1=N(0,1)$ and $f_2=N(0.5,1.25^2)$;
\item[] Case III: $f_1=N(0.5,0.75)$ and $f_2=N(0.5,1.25)$;
\item[] Case IV: $f_1=Logistic(0,1)$ and $f_2=Logistic(1,1)$;
\item[] Case V: $f_1=Logistic(0,1)$ and $f_2=Logistic(0.8,1.35)$;
\item[] Case VI: $f_1=Logistic(0.5,1)$ and $f_2=Logistic(0.5,1.5)$.
\end{enumerate}
The QTL affects only the location in Cases I and IV;  both the location and scale in Cases II and V;
and only the scale in Cases III and VI.
We consider the same four sample sizes $n$ and three values for
$d$ as  before.
To save space,
we present the results only for $n=100$ and 200;
the trends are similar for the other sample sizes.
The simulated powers of $R_n$, $R_n^*$, $KS$, and $AD$ for Cases I--III are presented in Table \ref{power1},
and those for Cases IV--VI are presented in Table \ref{power2}.
For a fair comparison, the critical values of
$R_n$, $R_n^*$, $KS$, and $AD$ are obtained by 10,000
repetitions under the null model.
All the power calculations are based on 1,000 repetitions.

\begin{center}
Tables \ref{power1}--\ref{power2} should be inserted here. 
\end{center}

Under the normal kernel, the performance trends of the simulation results displayed in Table \ref{power1} are as follows. When $f_1$ and $f_2$
have different means but the same variance, $R_n^*$ is the most  powerful of the four tests,
 and $R_n$ is always more powerful than $KS$ and $AD$.
In  contrast, when both the mean and variance of $f_1$ and $f_2$ are different,
$R_n$ is more powerful than the other three methods, and $R_n^*$ is better than $KS$ and $AD$. That is, if the QTL affects the mean but
not the variance, $R_n^*$ is more powerful at detecting
it than the other three methods, while if the QTL affects both the mean
and variance, $R_n$  is more powerful. All the powers increase as the sample size increases.
When $f_1$ and $f_2$ differ only in the variance, the powers of $R_n$ are far greater than those of the other three methods, whereas $R_n^*$ has almost no power. The power of $R_n$  increases  as the sample size increases,
while that of $R_n^*$ remains almost unchanged and is
close to the nominal type I error.
This implies that if the QTL affects
the variance but not the mean,
$R_n$ is more powerful at detecting it and $R_n^*$ seemingly fails.

Under the logistic kernel, the power performance trend is similar to that of the normal kernel. Hence, we omit the analysis.

Comparing the powers of $R_n$ in Tables \ref{power1} and \ref{power2},
we notice that  the power of $R_n$ under the logistic kernel
is  larger than that under the normal kernel in most cases
when other settings such as $d$ and $\theta$ are the same.
To explain this phenomenon,
we provide the Kullback-Leibler (KL) information with respect to the null model for all alternative models
in the last column of   Tables \ref{power1} and \ref{power2}.
It is seen that the KL information under Cases IV is larger than that under Case I.
This explains why the power of $R_n$  under Case IV  is larger.
The same interpretation is applicable to the comparison between Case VI and Case III.
We also note that the KL information under Cases V is close to that under Case II,
which explains why the powers of $R_n$ under these two cases are comparable.

\section{Real application}

In this section, we illustrate our test by analyzing
the data on the double haploid (DH) population of rice
in Example 11.3 of  \cite{Wu:2007}. The dataset
is available at  \url{http://www.buffalo.edu/~cxma/book/}.
We first give a brief background.
Two inbred lines, semi-dwarf IR64 and tall Azucena, were
crossed to generate an heterozygous $F_1$ progeny population.
Doubling the haploid chromosomes for the gametes of the
$F_1$ population led to 123 DH plants (\citealp{Huang:1997}).
 Such a DH population is equivalent to a backcross population because its marker segregation
follows 1:1 (\citealp{Huang:1997}).
These 123 DH plants were then genotyped for
135 RFLP and 40 isozyme and RAPD markers,
which represent a good coverage of 12 rice chromosomes (\citealp{Huang:1997}).

We use chromosome 1 for illustration. The
cumulative and pairwise map distances in centiMorgans for
18 markers on chromosome 1 are given in Table 11.4 of
\cite{Wu:2007}; this results in 17 intervals.
Table 11.4 of \cite{Wu:2007} also gives the sample size, sample mean, and sample variance
for the observations in each interval.
In  the analysis of \cite{Wu:2007},
$f_1$ and $f_2$ are assumed to be normal distributions.
To check the reasonability of this assumption,
we apply the Kolmogorov--Smirnov test for the normality of
the first sample and the fourth sample for each of the 17 intervals.
The results are summarized  in Table \ref{ks.pvalue}.
As we can see, all the $p$-values are greater than 0.3 in all 17 intervals.
This confirms that it is reasonable to assume that both $f_1$ and $f_2$ are normal distributions.

We next calculate $R_n$ and $R_n^*$ for each interval under the normality assumption on $f_1$ and $f_2$.
From (\ref{rep1})/\eqref{rep2},
we generate $N=100,000$ realizations and use them to  calculate the $p$-values of   $R_n$ and $R_n^*$.
Table \ref{real-data} shows their $p$-values  for the 17 intervals.
For comparison purposes, we also include the results from
$KS$, $AD$,  $R_n$, and $R_n^*$ under the logistic kernel for $f_1$ and $f_2$ for each interval.

We can see that both $R_n$ and $R_n^*$ under the normal kernel
suggest overwhelming evidence for the existence of a QTL in the last five intervals.
We also observe that the results  $R_n$ and $R_n^*$ under the normal kernel
are consistent with those from $KS$ and $AD$, and also those from $R_n$ and $R_n^*$ under the logistic kernel.
These findings may be confirmed by further experiments.
It is worth mentioning that in the interval RG173--RZ276, the  $p$-value (0.088)
of $R_n$ is much smaller than that (0.684) of $R_n^*$. At the 10\% significance level,
$R_n$ declares the existence of a QTL in this interval,
while $R_n^*$ fails. Since $R_n^*$ is designed to detect a QTL effect in only the mean,
while $R_n$ is able to detect a QTL effect in either the mean, the variance or both,
we believe that there exists a QTL effect in  only the variance  in this interval.

\begin{center}
Tables \ref{ks.pvalue}--\ref{real-data} should be inserted here.
\end{center}

\section{Discussion}

In practice, a QTL effect in the variance  may be  more crucial
than a QTL effect in the mean (\citealp{Korol:1996}).
In this paper, under the location-scale distribution family,
we studied the asymptotic properties of the MLEs of the unknown parameters
and the LRT for the existence of a QTL effect
under two general setups:
1) $f_1$ and $f_2$ may have different locations and/or scales,
and 2) $f_1$ and $f_2$ have the same unknown scale.
Our  theoretical results do not rely on the assumption
that the parameter space for the locations is bounded.
Explicit representations for the limiting distributions are  presented,
which facilitates the determination of the critical values.
These results enrich the literature and strengthen existing results on
QTL mapping in genetic backcross studies.
Our simulation studies show that
the explicit representations of the limiting distributions
result in LRTs with more accurate type I error rates than those based on
the approximation method of \cite{Davies:1987}.
Further, the LRT in both location and scale is uniformly more powerful than
the nonparametric tests for the homogeneity of distributions in multiple samples.

The results  in this paper are obtained under the assumption  that there is no double recombination.
When double recombination does occur,
 the data in each of the four groups are from a mixture distribution in both location and scale.
In this situation, the log-likelihood function  is unbounded, and hence the MLEs of the unknown parameters are not well defined.
We suggest  adding a penalty  on the scales, leading to a bounded penalized likelihood (\citealp{Chen:2008}).
The LRT based on the penalized likelihood can be constructed accordingly.
We expect that this new LRT will  have  similar properties to those in Theorem \ref{lrt-mean-var};
we leave this for future research.

\section*{Acknowledgments}
The authors thank the editor, associate editor, and two referees for constructive comments
and suggestions that led to significant improvements in the paper.
Dr. Li's research is supported in part by NSERC Grant RGPIN-2015-06592. Dr. Liu and Dr. Pu's research is supported by grants from the National Natural Science Foundation of China (11801359, 11771144, 11771145, 11971300),  the 111 Project (B14019) and Shanghai National Natural Science Foundation (19ZR1420900).

\section*{Supporting information}
Additional information for this article is available online as a supplementary material, which contains
the proof of Lemma \ref{lem1}, the limiting distribution of the LRT under the whole genome,  the illustration of some examples that satisfy Conditions A1--A7, and some additional simulation results.

\section*{Corresponding author's address}
Yukun Liu, Key Laboratory of Advanced
Theory and Application in Statistics and Data
Science - MOE, School of Statistics, East China
Normal University, 3663 Zhongshanbei Street, Shanghai 200062, China.
Email: ykliu@sfs.ecnu.edu.cn.

\section*{Appendix}
\subsection*{Regularity conditions}

The asymptotic properties of the LRTs rely on regularity conditions
on $f(y;\mu,\sigma)$. We impose the following  regularity
conditions on $f(y;\mu, \sigma)$ in which the expectations are taken under the null distribution $f(y; \mu_0, \sigma_0)$.

\begin{description}
\item[A1. (Integrability)] 
$\int_{\mathbb{R}} |\log f(y;0,1)|  f(y;0,1) dy<\infty$.
\item[A2. (Smoothness)]
The support of $f(y;0,1)$ is $(-\infty,\infty)$,
 and $f(y;0,1)$ is three times continuously differentiable with respect to $y$.

\item[A3. (Identifiability)]
For any two mixing distribution
functions $\psi_1$ and $\psi_2$ with two supporting
points such that
$$
\int f(y;\mu,\sigma)\,d\psi_1(\mu,\sigma)=\int f(y;\mu,\sigma)\,d\psi_2(\mu,\sigma)
$$
for all $y$, we must have $\psi_1=\psi_2$.

\item[A4. (Uniform boundedness)]
There exists a function $g$ such that
$$
{\Big|}\dfrac{\partial^{(h+l)} f(y;0,1)/\partial
\mu^h\partial \sigma^l}{f(y;0,1)}{\Big|}^3\leq g(y), \quad
\text{for }h+l\leq 2,
$$
for all $y$,where $h$ and $l$ are two nonnegative integers,
and
$$
\int_{\mathbb{R}} g(y)  f(y;0,1) dy<\infty.
$$
Moreover, there exists a positive $\epsilon$ such that
$$
\sup_{\mu^2+|\sigma-1|^2\leq \epsilon}
{\Big|}\dfrac{\partial^{(h+l)} f(y;\mu,\sigma)/\partial
\mu^h\partial \sigma^l}{f(y;0,1)}{\Big|}^3\leq g(y),
 \quad
\text{for }h+l=3.
$$

\item[A5. (Positive definiteness)]
The covariance matrix of
$(T,U)$ is positive definite, where
$$
T=\frac{\partial f(y_{11}; 0, 1)/\partial \mu} {
f(y_{11},0,1)}~~\mbox{and}~~ U=\frac{\partial
f(y_{11}; 0, 1)/\partial \sigma} { f(y_{11};0,1)}.
$$
That is, $T$ and $U$ are linearly uncorrelated. 
\item[A6. (Tail condition)]
There exist positive constants
$v_0$, $v_1$, and $\beta$  with $\beta>1$ such that
$f(y;0,1)\leq \min \{v_0, v_1|y|^{-\beta}\}$.

\item[A7. (Upper bound function)]
There exist  a nonnegative function $s(y;\mu,\sigma)$
 and three positive numbers  $(a_0, b_0, \epsilon^*)$ with  $\epsilon^*<1$,
 such that
(1)   $1/\beta<a_0<1$ with $\beta$ in Condition A6,
(2)  $s(y;0,1)$ is continuous in $y$ and satisfies  $\int_{\mathbb{R}} |\log s(y;0,1)|  s(y;0,1) dy<\infty$ and $\lim_{y\to \infty}s(y;0,1)=0$,
and
(3) for $\sigma\in(0,\epsilon^*\sigma_0)$, the function of $y$
$s(y;0,\sigma)$ is uniformly bounded, $\int s(y;0,\sigma)dy<1$,
and
$$
f(y;0,\sigma)\leq
\left\{
\begin{array}{cc}
\sigma^{-1}s(y;0,\sigma),&\mbox{ if } |y|\leq \sigma^{1-a_0}\\
\sigma^{b_0}s(y;0,\sigma),&\mbox{ if } |y|>\sigma^{1-a_0}\\
\end{array}
\right..
$$
\end{description}

\subsection*{Two technical lemmas}

Lemma \ref{lem1} establishes the consistency
of the MLEs under the null model; this is the first step in the proof of Theorems \ref{lrt-mean-var}--\ref{lrt-mean}.
The lemma claims that any estimator of $(\theta,\mu_1,\mu_2,\sigma_1,\sigma_2)$
with a large likelihood value is consistent for $\mu_h$ and $\sigma_h$, $h=1,2$,
under the null model.
Since both $R_n$ and $R_n^*$ are invariant to location and scale transformations, we assume that $(\mu_0,\sigma_0)=(0,1)$.

\renewcommand{\thelem}{1}
\begin{lem}
\label{lem1}
Assume the conditions of Theorem \ref{lrt-mean-var}.
Let $(\bar\theta, \bar\mu_1, \bar\mu_2,\bar\sigma_1,\bar\sigma_2)$ be any estimator of
$(\theta, \mu_1, \mu_2,\sigma_1,\sigma_2)$ such that
\begin{eqnarray}
\label{key-cond-lemma1}
l_n(\bar\theta, \bar\mu_1, \bar\mu_2,\bar\sigma_1,\bar\sigma_2)-l_n(0.5,0,0,1,1)
\geq c
>-\infty
\end{eqnarray}
for some constant $c$ for all $n$.
Then under the null model $f(y;0,1)$, $\bar\mu_1=o_p(1)$, $\bar\mu_2=o_p(1)$, $\bar\sigma_1-1=o_p(1)$, and $\bar\sigma_2-1=o_p(1)$.
\end{lem}

The proof of Lemma \ref{lem1} is quite long and technically involved. For convenience
of presentation, we leave it to the supplementary material.

\renewcommand{\therem}{A.\arabic{rem}}
\setcounter{rem}{0}
\begin{rem}
By the definition of MLE,
\[
l_n(\hat \theta, \hat \mu_1, \hat \mu_2, \hat \sigma_1, \hat \sigma_2)\geq l_n(0.5,0,0,1,1).
\]
Hence the MLE $(\hat \theta, \hat \mu_1, \hat \mu_2,\hat  \sigma_1,\hat  \sigma_2)$
of $(\theta, \mu_1, \mu_2,\sigma_1,\sigma_2)$ automatically satisfies
Condition \eqref{key-cond-lemma1} with  $c=0$.
Therefore the estimators in Theorems \ref{lrt-mean-var}--\ref{localpower}
all satisfy this condition automatically.
\end{rem}

In the next
lemma, we strengthen the conclusion of Lemma \ref{lem1} by providing an order assessment.
For convenience of presentation,
we define some notation.
Let $$
T_{ij}=\frac{\partial f(y_{ij};0,1)/\partial \mu}{f(y_{ij};0,1)},~~~
U_{ij}=\frac{{\partial f(y_{ij};0,1)/\partial \sigma}}{f(y_{ij};0,1)},~~i=1,\ldots,4,~~j=1,\ldots,n_i.
$$
Further,   let
$$
\bm{a}_{ij}=
\begin{pmatrix}
T_{ij}\\ U_{ij}
\end{pmatrix}, ~~
\bm{a}_i
=\sum_{j=1}^{n_i}\bm{a}_{ij},~~
\bm{a}
=\sum_{i=1}^{4}\bm{a}_{i},  ~~
\mbox{ and }~~
\bm {A}=\begin{pmatrix}
\sigma^2_T  &\sigma_{TU}\\
\sigma_{TU} &\sigma^2_U
\end{pmatrix},
$$
where  $\sigma^2_T=Var (T_{i1}), \sigma^2_U=Var(U_{i1})$, and
$\sigma_{TU}=Cov(T_{i1},   U_{i1})$.
We define
\begin{equation}
\label{barm}
\bar m_1(\theta)=\theta\bar\mu_1+(1-\theta)\bar\mu_2,~~
\bar m_2(\theta)=\theta(\bar\sigma_1-1)+(1-\theta)(\bar\sigma_2-1),
\end{equation}
and $\bar{\bm{m}}(\theta)=\Big( \bar m_1(\theta), \bar m_2(\theta)\Big)^\tau. $

\renewcommand{\thelem}{2}
\begin{lem}
\label{lem2}
Assume the conditions of Lemma \ref{lem1}.
Then under the null model $f(y;0,1)$, $\bar\mu_1=O_p(n^{-1/2})$,
$\bar\mu_2=O_p(n^{-1/2})$, $\bar\sigma_1-1=O_p(n^{-1/2})$, and $\bar\sigma_2-1=O_p(n^{-1/2})$.
\end{lem}
\begin{proof}
Let
\begin{eqnarray*}
 &&\hspace{-0.4in}l_{n1}(\mu_1,\sigma_1)=\sum\limits_{j=1}^{n_1}\log f_1(y_{1j}), ~
 l_{n2}(\theta, \mu_1,\mu_2,\sigma_1,\sigma_2)=\sum\limits_{j=1}^{n_{2}}\log
 \{ \theta f_1(y_{2j})+(1-\theta)f_2(y_{2j}) \},\\
  &&\hspace{-0.4in}l_{n3}(\theta, \mu_1,\mu_2,\sigma_1,\sigma_2)=
\sum\limits_{j=1}^{n_{3}}\log \{ (1-\theta) f_1(y_{3j})
+\theta f_2(y_{3j}) \},~
 l_{n4}(\mu_2,\sigma_2)=\sum\limits_{j=1}^{n_{4}}\log f_2(y_{4j}).
\end{eqnarray*}
Then
$$
l_n(\theta,\mu_1, \mu_2,\sigma_1,\sigma_2)
=l_{n1}(\mu_1,\sigma_1)+l_{n2}(\theta, \mu_1,\mu_2,\sigma_1,\sigma_2)+
l_{n3}(\theta, \mu_1,\mu_2,\sigma_1,\sigma_2)+ l_{n4}(\mu_2,\sigma_2).
$$
Next we derive an upper bound for $l_n(\bar\theta, \bar\mu_1, \bar\mu_2,\bar\sigma_1,\bar\sigma_2)-l_n(0.5,0,0,1,1)$. Together
with the lower bound $c$, we get the order assessment of $\bar\mu_h$ and $\bar\sigma_h$, $h=1, 2$.

From  Lemma \ref{lem1}, we have the consistency $\bar\mu_h=o_p(1)$ and $\bar\sigma_h-1=o_p(1)$, $h=1,2$.
Applying a second-order Taylor expansion to $l_{n1}(\bar\mu_1,\bar\sigma_1)-l_{n1}(0,1)$ around $(0,1)$
and using the law of large numbers, we get
\begin{equation}
\label{first.group}
l_{n1}(\bar\mu_1,\bar\sigma_1)-l_{n1}(0,1)
=\left\{\bar{\bm{m}}(1)\right\}^\tau \bm{a}_1-\frac{n_1}{2}\left\{\bar{\bm{m}}(1)\right\}^\tau \bm{A} \left\{\bar{\bm{m}}(1)\right\} \{1+o_p(1)\}.
\end{equation}
Similarly,
\begin{equation}
\label{forth.group}
l_{n4}(\bar\mu_2,\bar\sigma_2)-l_{n4}(0,1)
=\left\{\bar{\bm{m}}(0)\right\}^\tau \bm{a}_4-\frac{n_4}{2}\left\{\bar{\bm{m}}(0)\right\}^\tau \bm{A} \left\{\bar{\bm{m}}(0)\right\} \{1+o_p(1)\}.
\end{equation}

We now find an upper bound for
$l_{n2}(\bar \theta, \bar \mu_1,\bar \mu_2,\bar \sigma_1,\bar \sigma_2)-l_{n2}(0.5, 0,0,1,1)$.
Write $$
l_{n2}(\bar \theta, \bar \mu_1,\bar \mu_2,\bar \sigma_1,\bar \sigma_2)-l_{n2}(0.5, 0,0,1,1)
=\sum_{j=1}^{n_2}\log(1+\delta_j )
$$
with
$$
\delta_j
=\frac{\bar\theta\left\{f(y_{2j};\bar\mu_1,\bar\sigma_1)-f(y_{2j};0,1)\right\}+(1-\bar\theta)\left\{f(y_{2j};\bar\mu_2,\bar\sigma_2)-f(y_{2j};0,1)\right\}}{f(y_{2j};0,1)}.
$$
By the inequality $\log(1+x)\leq x-x^2/2+x^3/3$, we have
\begin{equation}
\label{ln2.upper1}
l_{n2}(\bar \theta, \bar \mu_1,\bar \mu_2,\bar \sigma_1,\bar \sigma_2)-l_{n2}(0.5, 0,0,1,1)
\leq \sum_{j=1}^{n_2}\delta_j-\sum_{j=1}^{n_2}\delta_j^2/2+\sum_{j=1}^{n_2}\delta_j^3/3.
\end{equation}
Applying a first-order Taylor expansion to  $f(y_{2j};\bar\mu_h,\bar\sigma_h)$, $h=1,2$,
we find that
$$
\delta_{j}
=
\left\{\bar {\bm{m}}(\bar\theta)\right\}^\tau \bm{a}_{2j}+ \varepsilon_{nj}
$$
and
the remainder term $\varepsilon_n=\sum_{j=1}^{n_2}\varepsilon_{nj}$
satisfies
\begin{equation}
\label{error.order}
\varepsilon_n=O_p(n_2^{1/2})\sum_{h=1}^2\left\{\mu_h^2+(\bar\sigma_h-1)^2\right\}
=o_p(n)\left\{||\bar{\bm{m}}(0) ||^2+||\bar{\bm{m}}(1) ||^2\right\}.
\end{equation}
Here $||\bar{\bm{m}}(0) ||$ and $||\bar{\bm{m}}(1) ||$ denote the $L_2$ norms of $\bar{\bm{m}}(0)$
and $\bar{\bm{m}}(1)$, respectively.
Therefore, for the linear term in (\ref{ln2.upper1}),
we have
\begin{equation}
\label{ln2.linear}
\sum_{j=1}^{n_2}\delta_j= \left\{\bar {\bm{m}}(\bar\theta)\right\}^\tau \bm{a}_{2}+\varepsilon_n,
\end{equation}
where the order of $\varepsilon_n$ is assessed in (\ref{error.order}).
After some straightforward algebra,  we find that
 the quadratic and cubic terms in (\ref{ln2.upper1}) satisfy
\begin{eqnarray*}
\sum_{j=1}^{n_2}\delta_j^2&=& \sum_{j=1}^{n_2}
\left[\left\{\bar {\bm{m}}(\bar\theta)\right\}^\tau \bm{a}_{2j}
\right]^2
+O_p(\varepsilon_n),\\
\sum_{j=1}^{n_2}\delta_j^3&=& \sum_{j=1}^{n_2}
\left[\left\{\bar {\bm{m}}(\bar\theta)\right\}^\tau \bm{a}_{2j}
\right]^3
+O_p(\varepsilon_n).
\end{eqnarray*}
By the strong law of large numbers,
 the fact that $||\bar {\bm{m}}(\bar\theta) ||^2\leq ||\bar {\bm{m}}(0) ||^2+||\bar {\bm{m}}(1) ||^2$,
 and the order assessment of $\varepsilon_n$ in (\ref{error.order}),
we have
\begin{eqnarray}\label{ln2.quad}
\sum_{j=1}^{n_2}\delta_j^2&=&n_2\left\{\bar {\bm{m}}(\bar\theta)\right\}^\tau \bm{A} \left\{\bar {\bm{m}}(\bar\theta)\right\}
+o_p(n) \left\{||\bar {\bm{m}}(0) ||^2+||\bar {\bm{m}}(1) ||^2\right\},\\
\label{ln2.cubic}\sum_{j=1}^{n_2}\delta_j^3&=& o_p(n) \left\{||\bar {\bm{m}}(0) ||^2+||\bar {\bm{m}}(1) ||^2\right\}.
\end{eqnarray}

Combining (\ref{ln2.upper1})--(\ref{ln2.cubic}), we get the refined upper bound for
$l_{n2}(\bar\theta, \bar \mu_1,\bar \mu_2,\bar \sigma_1,\bar \sigma_2)-l_{n2}(0.5, 0,0,1,1)$ as follows:
\begin{eqnarray}
\nonumber
&&l_{n2}(\bar\theta, \bar \mu_1,\bar \mu_2,\bar \sigma_1,\bar \sigma_2)-l_{n2}(0.5, 0,0,1,1)\\
&\leq&\left\{\bar {\bm{m}}(\bar\theta)\right\}^\tau\bm{a}_{2}-\frac{n_2}{2}\left\{\bar {\bm{m}}(\bar\theta)\right\}^\tau \bm{A} \left\{\bar {\bm{m}}(\bar\theta)\right\} +o_p(n)\left\{||\bar {\bm{m}}(0) ||^2+||\bar {\bm{m}}(1) ||^2\right\}.\nonumber\\
\label{second.group}
\end{eqnarray}
Similarly,
\begin{eqnarray}
\nonumber
&&l_{n3}(\bar \theta, \bar \mu_1,\bar \mu_2,\bar \sigma_1,\bar \sigma_2)-l_{n3}(0.5, 0,0,1,1)\nonumber\\
&\leq&\left\{\bar {\bm{m}}(1-\bar\theta)\right\}^\tau\bm{a}_{3}-\frac{n_3}{2}\left\{\bar {\bm{m}}(1-\bar\theta)\right\}^\tau \bm{A} \left\{\bar {\bm{m}}(1-\bar\theta)\right\}\nonumber\\
&&+o_p(n)\left\{||\bar {\bm{m}}(0) ||^2+||\bar {\bm{m}}(1) ||^2\right\}.\label{third.group}
\end{eqnarray}
Combining (\ref{first.group}), (\ref{forth.group}), (\ref{second.group}), and (\ref{third.group}), we have
\begin{eqnarray}
c&\leq& l_{n}(\bar \theta, \bar \mu_1,\bar \mu_2,\bar \sigma_1,\bar \sigma_2)-l_{n}(0.5,0,0,1,1)\nonumber\\
&\leq& \left\{\bar{\bm{m}}(1)\right\}^\tau \bm{a}_1-\frac{n_1}{2}\left\{\bar{\bm{m}}(1)\right\}^\tau \bm{A} \left\{\bar{\bm{m}}(1)\right\} \{1+o_p(1)\}\nonumber\\
&&+ \left\{\bar {\bm{m}}(\bar\theta)\right\}^\tau\bm{a}_{2}-\frac{n_2}{2}\left\{\bar {\bm{m}}(\bar\theta)\right\}^\tau \bm{A} \left\{\bar {\bm{m}}(\bar\theta)\right\}\nonumber\\
&& +\left\{\bar {\bm{m}}(1-\bar\theta)\right\}^\tau\bm{a}_{3}-\frac{n_3}{2}\left\{\bar {\bm{m}}(1-\bar\theta)\right\}^\tau \bm{A} \left\{\bar {\bm{m}}(1-\bar\theta)\right\}\nonumber\\
&&+\left\{\bar{\bm{m}}(0)\right\}^\tau \bm{a}_4-\frac{n_4}{2}\left\{\bar{\bm{m}}(0)\right\}^\tau \bm{A} \left\{\bar{\bm{m}}(0)\right\} \{1+o_p(1)\}\label{upper.bound1}\\
&\leq& \left\{\bar{\bm{m}}(1)\right\}^\tau \bm{a}_1-\frac{n_1}{2}\left\{\bar{\bm{m}}(1)\right\}^\tau \bm{A} \left\{\bar{\bm{m}}(1)\right\} \{1+o_p(1)\}\nonumber\\
&&+\left\{\bar{\bm{m}}(0)\right\}^\tau \bm{a}_4-\frac{n_4}{2}\left\{\bar{\bm{m}}(0)\right\}^\tau \bm{A} \left\{\bar{\bm{m}}(0)\right\} \{1+o_p(1)\}+O_p(1).
\label{upper.bound2}
\end{eqnarray}
By Condition A5, $\bm{A}$ is positive definite. Then the
upper bound of $ l_{n}(\bar \theta, \bar \mu_1,\bar \mu_2,\bar \sigma_1,\bar \sigma_2)-l_{n}(0.5,0,0,1,1)$ in (\ref{upper.bound2})  has order $O_p(1)$.
Together with the lower bound $c$, this implies that
$$
\bar\mu_1=O_p(n^{-1/2}),~
\bar\sigma_1-1=O_p(n^{-1/2}),~
\bar\mu_2=O_p(n^{-1/2}),~
\bar\sigma_2-1=O_p(n^{-1/2}).
$$
Any values of $(\bar\mu_1,\bar\sigma_1-1,\bar\mu_2,\bar\sigma_2-1)$ outside
this range will violate
the inequality.
This completes the proof.
\end{proof}

\subsection*{Proof of Theorem \ref{lrt-mean-var}}

\noindent{\it Proof of Part (i).}

By the definition of the MLE,
we  have $l_n(\hat\theta,\hat\mu_1,\hat\mu_2,\hat\sigma_1,\hat\sigma_2)-l_n(0.5, 0, 0, 1, 1) \geq0$.
Hence, Condition (\ref{key-cond-lemma1}) is satisfied. Then
applying the results in Lemma \ref{lem2}, we obtain the results in  Part (i).

\vspace{0.1in}
\noindent{\it Proof of Part (ii).}

Note that
\begin{eqnarray}
\label{LLR}
R_n=2\left\{l_n(\hat\theta,\hat\mu_1,\hat\mu_2,\hat\sigma_1,\hat\sigma_2 )-l_n(0.5,\hat\mu_0,\hat\mu_0,\hat\sigma_0,\hat\sigma_0) \right\}=R_{1n}-R_{2n},
\end{eqnarray}
where $$R_{1n}=2\left\{l_n(\hat\theta,\hat\mu_1,\hat\mu_2,\hat\sigma_1,\hat\sigma_2 )-l_n(0.5,0,0,1,1)\right\}$$
and $$R_{2n}=2\left\{ l_n(0.5,\hat\mu_0,\hat\mu_0,\hat\sigma_0,\hat\sigma_0)-l_n(0.5,0,0,1,1)\right\}. $$

Applying some of the classical results for regular models (\citealp{Serfling:1980}), we have
\begin{equation}
\label{LLR2}
R_{2n}=\bm {a}^\tau (n\bm{A})^{-1}\bm {a}+o_p(1).
\end{equation}

Next we use a sandwich method to find the approximation of $R_{1n}$.
We proceed in two steps. In Step 1, we derive an upper bound for $R_{1n}$
and in Step 2, we argue that the upper bound is achievable.

By Part (i), the results in (\ref{upper.bound1}) are applicable to $R_{1n}$.
Hence,
\begin{eqnarray}
R_{1n}&=&2\left\{l_n(\hat\theta,\hat\mu_1,\hat\mu_2,\hat\sigma_1,\hat\sigma_2 )-l_n(0.5,0,0,1,1)\right\}\nonumber\\
&\leq& 2\left\{\hat{\bm{m}}(1)\right\}^\tau \bm{a}_1-n_1\left\{\hat{\bm{m}}(1)\right\}^\tau \bm{A} \left\{\hat{\bm{m}}(1)\right\} \{1+o_p(1)\}\nonumber\\
&&+ 2\left\{\hat {\bm{m}}(\hat\theta)\right\}^\tau\bm{a}_{2}-n_2\left\{\hat {\bm{m}}(\hat\theta)\right\}^\tau \bm{A} \left\{\hat {\bm{m}}(\hat\theta)\right\}\nonumber\\
&& +2\left\{\hat {\bm{m}}(1-\hat\theta)\right\}^\tau\bm{a}_{3}-n_3\left\{\hat {\bm{m}}(1-\hat\theta)\right\}^\tau \bm{A} \left\{\hat {\bm{m}}(1-\hat\theta)\right\}\nonumber\\
&&+2\left\{\hat{\bm{m}}(0)\right\}^\tau \bm{a}_4-n_4\left\{\hat{\bm{m}}(0)\right\}^\tau \bm{A} \left\{\hat{\bm{m}}(0)\right\} \{1+o_p(1)\}.\label{upper.bound3}
\end{eqnarray}
Here $\hat {\bm{m}}(\theta)$  is defined similarly to (\ref{barm}) with
$(\hat\mu_1,\hat\mu_2,\hat\sigma_1,\hat\sigma_2)$
in place of   $(\bar\mu_1,\bar\mu_2,\bar\sigma_1,\bar\sigma_2)$.

To simplify the upper bound in (\ref{upper.bound3}), let
\begin{equation*}
\hspace{-0.1in}
\hat\bbeta_1=\frac{\hat {\bm{m}}(1)+\hat {\bm{m}}(0)}{2}
=
\begin{pmatrix}
\frac{\hat\mu_1+\hat\mu_2}{2}\\
\frac{\hat\sigma_1+\hat\sigma_2-2}{2}
\end{pmatrix}
\mbox{ and }
\hat\bbeta_2=\frac{\hat {\bm{m}}(1)-\hat {\bm{m}}(0)}{2}
=\begin{pmatrix}
\frac{\hat\mu_1-\hat\mu_2}{2}\\
\frac{\hat\sigma_1-\hat\sigma_2}{2}
\end{pmatrix}
.
\end{equation*}
Then
\begin{equation}
\label{hatm}
\hat {\bm{m} }(\theta)=\hat\bbeta_1+2(\theta-0.5)\hat\bbeta_2.
\end{equation}
Substituting (\ref{hatm}) into (\ref{upper.bound3}) and using Part (i) and the condition that $n_i/n\to p_i$, $i=1,2,3,4$ with
$$(p_1,p_2,p_3,p_4)=\left(\frac{1-r}{2}, \frac{r}{2}, \frac{r}{2}, \frac{1-r}{2} \right),$$
we have
\begin{eqnarray}
R_{1n}&=&2\left\{l_n(\hat\theta,\hat\mu_1,\hat\mu_2,\hat\sigma_1,\hat\sigma_2 )-l_n(0.5,0,0,1,1)\right\}\nonumber\\
&\leq& 2\hat\bbeta_1^\tau \bm{a}-n\hat\bbeta_1^\tau \bm{A} \hat\bbeta_1
\{1+o_p(1)\}+ 2\hat\bbeta_2^\tau \bm{b}(\hat\theta)-n\tau(\hat\theta)\hat\bbeta_2^\tau \bm{A} \hat\bbeta_2 \nonumber\\
 &\leq&\bm{a}^\tau(n\bm{A})^{-1}\bm{a}+\left\{\bm{b}(\hat\theta)\right\}^\tau\left\{n\tau(\hat\theta)\bm{A} \right\}^{-1}\left\{\bm{b}(\hat\theta)\right\}+o_p(1)\nonumber\\
 &\leq&\bm{a}^\tau(n\bm{A})^{-1}\bm{a}+\sup_{\theta\in[0,1]}\left[\left\{\bm{b}(\theta)\right\}^\tau\left\{n\tau(\theta)\bm{A} \right\}^{-1}\left\{\bm{b}(\theta)\right\}\right]+o_p(1) \label{upper.bound4},
\end{eqnarray}
where
$$
\bm{b}(\theta)=\bm{a}_1-\bm{a}_4+(2\theta-1)(\bm{a}_2-\bm{a}_3)
\mbox{  and  }
\tau(\theta)=1+4r\theta(\theta-1).
$$

Next, we show that the upper bound in (\ref{upper.bound4}) for  $R_{1n}$ is
achievable.
Let
$$
\tilde\theta=\arg\max_{\theta\in[0,1]}\left[
\left\{\bm{b}(\theta)\right\}^\tau\left\{n\tau(\theta)\bm{A} \right\}^{-1}\left\{\bm{b}(\theta)\right\}
\right]
$$
and $(\tilde\mu_1,\tilde\mu_2,\tilde\sigma_1,\tilde\sigma_2)$
be determined by
$$
\begin{pmatrix}
\frac{\tilde\mu_1+\tilde\mu_2}{2}\\
\frac{\tilde\sigma_1+\tilde\sigma_2-2}{2}
\end{pmatrix}
=(n\bm{A})^{-1}\bm{a}
~
\mbox{ and }
~
\begin{pmatrix}
\frac{\tilde\mu_1-\tilde\mu_2}{2}\\
\frac{\tilde\sigma_1-\tilde\sigma_2}{2}
\end{pmatrix}
=
\left\{n\tau(\tilde\theta)\bm{A} \right\}^{-1}
\left\{\bm{b}(\tilde\theta)\right\}.
$$
Note that it is easy to verify that $(\tilde\mu_1,\tilde\mu_2,\tilde\sigma_1,\tilde\sigma_2)$ exists and
$$
\tilde \mu_h=O_p(n^{-1/2}),~~
\tilde\sigma_h-1=O_p(n^{-1/2}),~~h=1,2.
$$
With this order assessment and applying a second-order Taylor expansion, we
have
\begin{eqnarray}
\nonumber
R_{1n}&\geq&2\left\{l_n(\tilde\theta, \tilde\mu_1,\tilde\mu_2,\tilde\sigma_1,\tilde\sigma_2)
-l_n(0.5,0,0,1,1)\right\}\\
&=& \bm{a}^\tau(n\bm{A})^{-1}\bm{a}+
\sup_{\theta\in[0,1]}\left[
\left\{\bm{b}(\theta)\right\}^\tau\left\{n\tau(\theta)\bm{A} \right\}^{-1}\left\{\bm{b}(\theta)\right\}
\right] +o_p(1).
\label{lower.r1n}
\end{eqnarray}
Combining \eqref{upper.bound4} and \eqref{lower.r1n} leads to
$$
R_{1n}= \bm{a}^\tau(n\bm{A})^{-1}\bm{a}+
\sup_{\theta\in[0,1]}\left[
\left\{\bm{b}(\theta)\right\}^\tau\left\{n\tau(\theta)\bm{A} \right\}^{-1}\left\{\bm{b}(\theta)\right\}
\right] +o_p(1).
$$
With (\ref{LLR}) and (\ref{LLR2}), we further have that
\begin{equation}
\label{lower.rn}
R_n= \sup_{\theta\in[0,1]}\left[
\left\{\bm{b}(\theta)\right\}^\tau\left\{n\tau(\theta)\bm{A} \right\}^{-1}\left\{\bm{b}(\theta)\right\}
\right] +o_p(1).
\end{equation}

It  can be verified that the process
$\left\{\bm{b}(\theta)\right\}^\tau\left\{n\tau(\theta)\bm{A} \right\}^{-1}\left\{\bm{b}(\theta)\right\}$
converges weakly to the process $Z_1^2(\theta)+Z_2^2(\theta)$, where $Z_1(\theta)$ and $Z_2(\theta)$
are defined in (\ref{zix}) (see \citealp{Kim:2013}).
Hence,
$$
R_n\to R=\sup_{0\leq \theta\leq 1}\{ Z^2_1(\theta)+Z^2_2(\theta)\}
$$
in distribution, as $n\rightarrow \infty$. This completes the proof of Part (ii).
\vspace{0.1in}

\noindent{\it Proof of (\ref{rep1}).}

Let
$
c_1(\theta)=
\sqrt{1-r}{\big/}\sqrt{1+ 4r\theta (\theta-1)}
$ and
$c_2(\theta)=\sqrt{r}(2\theta-1){\big/}\sqrt{1+ 4r\theta(\theta-1)}.
$
Recall the forms of $Z_1(\theta)$ and $Z_2(\theta)$
defined in (\ref{zix}).
Then $R$ can be written as
\begin{align}
R
&=Z^2_1(\theta)+Z^2_2(\theta)
= \{c_1(\theta)z_{11}+c_2(\theta)z_{12}\}^2
+ \{ c_1(\theta)z_{21} + c_2(\theta)z_{22}\}^2\nonumber\\
&=\begin{pmatrix}
c_1(\theta), &c_2(\theta)
\end{pmatrix}
\begin{pmatrix}
z_{11}^2+z^2_{21}  &z_{11}z_{12}+z_{21}z_{22}\\
z_{11}z_{12}+z_{21}z_{22} &z_{12}^2+z^2_{22}
\end{pmatrix}
\begin{pmatrix}
c_1(\theta)\\
c_2(\theta)
\end{pmatrix},
\label{square.sum}
\end{align}
where  $z_{11}$, $z_{12}$, $z_{21}$,
and $z_{22}$ are four independent standard normal
variables.

Similarly to Lemma 3 of \cite{Zhang:2008}, we have
$$
\left\{\big(c_1(\theta),c_2(\theta)\big):0\leq \theta \leq 1\right\}=\{(x_1,x_2):
x_1^2+x_2^2=1, x_1\geq\sqrt{1-r}\}.
$$
Let
$$
\bm{W}=
\begin{pmatrix}
z_{11}^2+z^2_{21}  &z_{11}z_{12}+z_{21}z_{22}\\
z_{11}z_{12}+z_{21}z_{22} &z_{12}^2+z^2_{22}
\end{pmatrix}, ~~\mathscr{B}=\{(x_1,x_2):
x_1^2+x_2^2=1, x_1\geq\sqrt{1-r}\}.
$$
Then, from \eqref{square.sum} we have
\begin{equation}
\label{form.R}
R=\sup_{(x_1,x_2)\in\mathscr{B}}(x_1,x_2)\bm{W}
(x_1,x_2)^\tau.
\end{equation}

To help us find the maximum of $(x_1,x_2)\bm{W}
(x_1,x_2)^\tau$, we make the following polar transformations.
Let
$$
\begin{cases}
(z_{11}-z_{22}){\big/}\sqrt{2}=\rho_1\sin U_1\\
(z_{12}+z_{21}){\big/}\sqrt{2}=\rho_1\cos U_1
\end{cases} \quad \mbox{and} \quad~~~
\begin{cases}
(z_{11}+z_{22}){\big/}\sqrt{2}=\rho_2\cos U_2\\
(z_{21}-z_{12}){\big/}\sqrt{2}=\rho_2\sin U_2
\end{cases},
$$
where   $\rho_1^2$ with $\rho_1>0$, $\rho_2^2$ with $\rho_2>0$, $U_1$, and $U_2$ are four independent
random variables with $\rho_1^2$ and $\rho_2^2$ from a $\chi_2^2$ distribution, and
$U_1$ and $U_2$ from a uniform distribution on $[-3\pi /4, 5\pi/4]$.

It can be verified that the two eigenvalues of $\bm{W}$ are
$\lambda_1=(1/2)(\rho_1+\rho_2)^2$ and
$\lambda_2=(1/2)(\rho_1-\rho_2)^2$, respectively.
Further, $\bm{W}$ can be decomposed as
\begin{equation}
\label{decomposition}
\bm{W}=\bm{P}^\tau
\begin{pmatrix}
\lambda_1 &0\\
0  & \lambda_2
\end{pmatrix}
\bm{P},
\end{equation}
where $$\bm{P}=\begin{pmatrix}
\cos\{(U_1+U_2)/2-\pi/4\}
&-\sin\{(U_1+U_2)/2-\pi/4\}\\
\sin\{(U_1+U_2)/2-\pi/4\}
&\cos\{(U_1+U_2)/2-\pi/4\}
\end{pmatrix}.$$

Since $\bm{P}$ is an orthogonal transformation and   $ x_1^2+x_2^2=1$,
it follows that $(x_1, x_2)\bm{P}^\tau \bm{P}(x_1, x_2)^\tau = x_1^2+x_2^2=1$. Therefore,
we can write $(x_1, x_2)\bm{P}^\tau =(\cos\alpha,\sin\alpha)$ with $\alpha\in[-\pi, \pi]$.
From \eqref{form.R} and \eqref{decomposition}, we have
$$
R=\sup_{\alpha\in\mathscr{F}}
(\lambda_1\cos^2\alpha+\lambda_2\sin^2\alpha)
=\sup_{\alpha\in\mathscr{F}}\{(\lambda_1-\lambda_2)\cos^2\alpha+\lambda_2\},
$$
where
$$
\mathscr{F}=\{\alpha:
\cos\{(U_1+U_2)/2-\pi/4\}\cos\alpha+
\sin\{(U_1+U_2)/2-\pi/4\}\sin\alpha\geq\sqrt{1-r}\}.$$
Recall that $\eta=(U_1+U_2)/2-\pi/4$, which satisfies $\eta\in[-\pi,\pi]$.
Hence,
$$
\mathscr{F}=\{\alpha:
\cos(\alpha-\eta)\geq\sqrt{1-r}\}.$$
Therefore, after some simple
analysis, (\ref{rep1}) follows.
This completes the proof of Theorem \ref{lrt-mean-var}.

\subsection*{Proof of Theorem \ref{thm3}}

\noindent{\it Proof of Part (i).}

Assume $(\mu_0,\sigma_0)=(0,1)$, then the local alternative \eqref{alternative} is equivalent to
\begin{equation*}
H_{A}^{n}: \theta=\theta_0, ~
\begin{pmatrix}\mu_1\\ \mu_2\end{pmatrix}=
\begin{pmatrix}-n^{-1/2}\delta_{\mu}/\sigma_0\\
n^{-1/2}\delta_{\mu}/\sigma_0\end{pmatrix}, ~
\begin{pmatrix}\sigma_1\\ \sigma_2\end{pmatrix}=
\begin{pmatrix}1-n^{-1/2}\delta_{\sigma}/\sigma_0\\
1+n^{-1/2}\delta_{\sigma}/\sigma_0\end{pmatrix}.
\end{equation*}
By \eqref{lower.rn}, $R_n=\sup_{\theta\in[0,1]}\{R_n(\theta)\}+o_p(1)$ under the null model, where
$$
R_n(\theta)=
\left\{\bm{b}(\theta)\right\}^\tau\left\{n\tau(\theta)\bm{A} \right\}^{-1}\left\{\bm{b}(\theta)\right\}.
$$
Next we derive the limiting distribution of $R_n$ under $H_{A}^n$.

Let $\Lambda_n=\Lambda_{n1}+\Lambda_{n2}+\Lambda_{n3}+\Lambda_{n4}$,
where
\begin{align*}
&\Lambda_{n1}=l_{n1}(\mu_1,\sigma_1)-l_{n1}(0,1),~
\Lambda_{n2}=l_{n2}(\theta_0,\mu_1,\mu_2,\sigma_1,\sigma_2)-
l_{n3}(0.5,0,0,1,1),\\
&\Lambda_{n3}=l_{n3}(\theta_0,\mu_1,\mu_2,\sigma_1,\sigma_2)-
l_{n3}(0.5,0,0,1,1), ~\mbox{and}~
\Lambda_{n4}=l_{n4}(\mu_1,\sigma_1)-l_{n4}(0,1).
\end{align*}
Using the second Taylor  expansion, under the null, we have
\begin{align*}
&\Lambda_{n1}=
-n^{-1/2}\sum_{j=1}^{n_1}(\Delta_{\mu}T_{1j}+\Delta_{\sigma}U_{1j})
-0.5p_1(\Delta_{\mu},\Delta_{\sigma}){\bm A}
\begin{pmatrix}\Delta_{\mu}\\ \Delta_{\sigma}\end{pmatrix}+o_p(1),\\
&\Lambda_{n4}=
n^{-1/2}\sum_{j=1}^{n_4}(\Delta_{\mu}T_{4j}+\Delta_{\sigma}U_{4j})
-0.5p_4(\Delta_{\mu},\Delta_{\sigma}){\bm A}
\begin{pmatrix}\Delta_{\mu}\\ \Delta_{\sigma}\end{pmatrix}+o_p(1),
\end{align*}
where $\Delta_{\mu}=\delta_{\mu}/\sigma_0$ and
$\Delta_{\sigma}=\delta_{\sigma}/\sigma_0$.
Similarly,
\begin{align*}
\Lambda_{n2}=&
n^{-1/2}\sum_{j=1}^{n_2}\{m_1(\theta_0)T_{2j}+m_2(\theta_0)U_{2j}\}
\\
&
-0.5p_2\big(m_1(\theta_0),m_2(\theta_0)\big){\bm A}
\begin{pmatrix}m_1(\theta_0)\\ m_2(\theta_0)\end{pmatrix}+o_p(1),
\end{align*}
and
\begin{align*}
\Lambda_{n3}=&~
n^{-1/2}\sum_{j=1}^{n_3}\{m_1(1-\theta_0)T_{3j}+m_2(1-\theta_0)U_{3j}\}
\\
&~
-0.5p_3\big(m_1(1-\theta_0),m_2(1-\theta_0)\big){\bm A}
\begin{pmatrix}m_1(1-\theta_0)\\ m_2(1-\theta_0)\end{pmatrix}+o_p(1),
\end{align*}
where $m_1(\theta)=(1-2\theta)\Delta_{\mu}$, and $m_2(\theta)=(1-2\theta)\Delta_{\sigma}$.

By the central limit theorem, we get
$\Lambda_n\stackrel{d}\rightarrow N(-0.5c_0,c_0)$ under the null, where
\begin{align*}
c_0=&~2p_1\left(\Delta_{\mu},\Delta_{\sigma}\right){\bm A}
\begin{pmatrix}\Delta_{\mu}\\ \Delta_{\sigma}\end{pmatrix}
+2p_2\big(m_1(\theta_0),m_2(\theta_0)\big){\bm A}
\begin{pmatrix}m_1(\theta_0)\\ m_2(\theta_0)\end{pmatrix}.
\end{align*}
Therefore, the local alternative $H_{A}^n$ is contiguous to the null distribution (\citealp{Le Cam:1990} and Example 6.5 of \citealp{van der Vaart:2000}). By Le Cam's contiguity theory, the limiting distribution of $R_n(\theta)$ under $H_{A}^n$ is determined by the joint limiting distribution of $\{n\tau(\theta){\bm A}\}^{-1/2}\bm b(\theta)$ and $\Lambda_n$ under the null model.

By the central limit theorem and Slutsky's theorem, the joint limiting distribution of $\{n\tau(\theta){\bm A}\}^{-1/2}\bm b(\theta)$ and $\Lambda_n$ under the null model is multivariate normal
$$
\mathcal{N}_3\left(
\begin{pmatrix}
{\bf 0}\\
-0.5c_0
\end{pmatrix},~
\begin{pmatrix}
\bm I_2 &\bm\rho_{\theta_0}(\theta)\\
\bm\rho_{\theta_0}^\tau(\theta)  &c_0
\end{pmatrix}
\right) ~{\rm with}~
\bm\rho_{\theta_0}(\theta)=\bm\Delta_{\theta_0}(\theta)
\begin{pmatrix}
\Delta_{\mu}\\
\Delta_{\sigma}
\end{pmatrix},
$$
where $\bm \Delta_{\theta_0}(\theta)=-\{1+2r(2\theta_0\theta-\theta_0-\theta)\}\{\tau(\theta)\}^{-\frac{1}{2}}
\bm A^{\frac{1}{2}}$.
By Le Cam's third lemma (\citealp{van der Vaart:2000}), we have under $H_{A}^{n}$,
$$
\{n\tau(\theta){\bm A}\}^{-1/2}\bm b(\theta)
\stackrel{d}\rightarrow \mathcal{N}_2
\left(\bm\rho_{\theta_0}(\theta), \bm I_2\right).
$$
Further, we have $R_n(\theta)\stackrel{d}\rightarrow
\chi_2^2\big(\bm\rho_{\theta_0}^\tau(\theta)\bm\rho_{\theta_0}(\theta)\big)$ under $H_{A}^n$. Because
$R_n=\sup_{\theta\in[0,1]}\{R_n(\theta)\}+o_p(1)$ under the null model,
by Le Cam's first lemma (\citealp{van der Vaart:2000}),
$R_n=\sup_{\theta\in[0,1]}\{R_n(\theta)\}+o_p(1)$ still holds under the local alternative $H_{A}^n$. Therefore, the limiting distribution
of $R_n$ under $H_{A}^n$ is $\sup_{\theta\in[0,1]}\left\{\chi_2^2\big(\bm \rho_{\theta_0}^\tau(\theta)\bm\rho_{\theta_0}(\theta)\big)\right\}$.

\vspace{0.1in}
\noindent{\it Proof of Part (ii).}

The proof of this part is similar to  that of  Part
(i)  and hence is omitted.

\begin{figure}[bt]
\centering
\includegraphics[width=10cm]{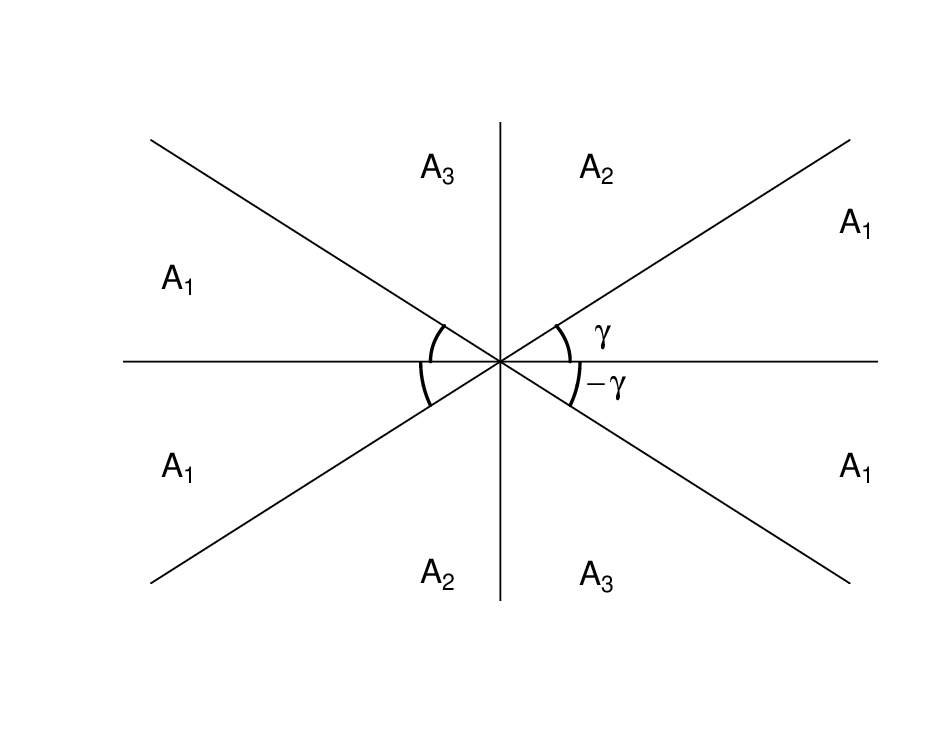}
\vspace{-4.5em}
\caption{Graphical presentation of the sets $A_1$, $A_2$, and $A_3$ defined in
(5).}
\label{figure1}
\end{figure}

\begin{table}[!ht]
\setlength{\abovecaptionskip}{-0.03pt}
\tabcolsep=1.36mm
\centering
\caption{\label{typeI1} Empirical type I errors (\%)  of $R_n$, $R_n^*$, $KS$, and $AD$.
Here ``Ours" means our method based on the explicit representations
and ``Davies" means Davies's approximation method.
The random samples are generated from model
\eqref{model} with $f_1=f_2=N(0,1)$ and $f_1=f_2=Logistic(0,1)$, respectively.}
\begin{tabular}{ccccccccccccc}
\hline\hiderowcolors
\multicolumn{13}{c}{$f_1=f_2=N(0,1)$}\\
\cmidrule(lr){1-13}
$\alpha=5\%$   &\multicolumn{6}{c}{$d=5$}
 &\multicolumn{6}{c}{$d=20$}\\
\cmidrule(lr){1-1} \cmidrule(lr){2-7}  \cmidrule(lr){8-13}
$n$ &\multicolumn{2}{c}{$R_n$}&\multicolumn{2}{c}{$R_n^*$}&$KS$ &$AD$
    &\multicolumn{2}{c}{$R_n$}&\multicolumn{2}{c}{$R_n^*$}&$KS$ &$AD$\\
&Ours &Davies &Ours &Davies& &
    &Ours &Davies &Ours &Davies& & \\
50  &5.77 &6.42 &5.18 &5.71 &3.12 &5.31
    &5.77 &6.74 &5.88 &6.84 &4.30 &5.10\\
100 &5.52 &6.04 &5.32 &5.87 &3.48 &5.17
    &5.70 &6.68 &5.46 &6.46 &5.19 &5.41\\
200 &5.08 &5.68 &5.19 &5.59 &3.89 &4.81
    &5.36 &6.15 &4.93 &5.80 &4.64 &4.59\\
300 &4.62 &5.13 &4.95 &5.48 &4.18 &4.83
    &4.51 &5.26 &4.92 &5.78 &4.93 &5.00\\
\cmidrule(lr){1-13}
$\alpha=1\%$   &\multicolumn{6}{c}{$d=5$}
 &\multicolumn{6}{c}{$d=20$}\\
\cmidrule(lr){1-1} \cmidrule(lr){2-7}  \cmidrule(lr){8-13}
$n$ &\multicolumn{2}{c}{$R_n$}&\multicolumn{2}{c}{$R_n^*$}&$KS$ &$AD$
    &\multicolumn{2}{c}{$R_n$}&\multicolumn{2}{c}{$R_n^*$}&$KS$ &$AD$\\
&Ours &Davies &Ours &Davies& &
    &Ours &Davies &Ours &Davies& & \\
50  &1.30 &1.50 &1.30 &1.50 &0.47 &0.77
    &1.20 &1.40 &1.20 &1.38 &0.57 &0.91\\
100 &1.05 &1.17 &1.13 &1.30 &0.40 &0.90
    &1.34 &1.54 &1.33 &1.59 &0.94 &1.22\\
200 &0.92 &1.05 &1.13 &1.34 &0.56 &0.87
    &1.10 &1.21 &0.99 &1.24 &0.92 &0.98\\
300 &0.91 &0.97 &0.90 &1.04 &0.69 &0.84
    &1.00 &1.14 &0.99 &1.16 &0.93 &1.05\\
\hline
\hline
\multicolumn{13}{c}{$f_1=f_2=Logistic(0,1)$}\\
\cmidrule(lr){1-13}
$\alpha=5\%$   &\multicolumn{5}{c}{$d=5$} &\multicolumn{5}{c}{$d=20$}\\
\cmidrule(lr){1-1} \cmidrule(lr){2-7}  \cmidrule(lr){8-13}
$n$ &\multicolumn{2}{c}{$R_n$}&\multicolumn{2}{c}{$R_n^*$}&$KS$ &$AD$
    &\multicolumn{2}{c}{$R_n$}&\multicolumn{2}{c}{$R_n^*$}&$KS$ &$AD$\\
&Ours &Davies &Ours &Davies& &
    &Ours &Davies &Ours &Davies& & \\
50  &5.80 &6.22 &5.25 &5.72 &3.09 &5.34
    &5.89 &6.86 &5.67 &6.77 &3.96 &4.76\\
100 &5.09 &5.63 &5.10 &5.68 &3.59 &5.22
    &5.04 &5.95 &4.91 &5.81 &4.90 &5.21\\
200 &4.99 &5.49 &4.92 &5.58 &4.06 &5.11
    &5.11 &6.04 &5.19 &6.14 &5.11 &5.09\\
300 &5.06 &5.53 &4.97 &5.43 &4.59 &4.97
    &4.99 &5.87 &4.95 &6.03 &5.07 &5.23\\
\cmidrule(lr){1-13}
$\alpha=1\%$   &\multicolumn{6}{c}{$d=5$}
&\multicolumn{6}{c}{$d=20$}\\
\cmidrule(lr){1-1} \cmidrule(lr){2-7}  \cmidrule(lr){8-13}
$n$ &\multicolumn{2}{c}{$R_n$}&\multicolumn{2}{c}{$R_n^*$}&$KS$ &$AD$
    &\multicolumn{2}{c}{$R_n$}&\multicolumn{2}{c}{$R_n^*$}&$KS$ &$AD$\\
&Ours &Davies &Ours &Davies& &
    &Ours &Davies &Ours &Davies& & \\
50  &1.33 &1.52 &1.12 &1.29 &0.37 &0.61
    &1.13 &1.37 &1.16 &1.48 &0.51 &0.69\\
100 &1.00 &1.19 &0.97 &1.14 &0.52 &0.90
    &0.98 &1.15 &0.91 &1.06 &0.76 &0.96\\
200 &1.21 &1.28 &0.95 &1.07 &0.66 &0.97
    &1.24 &1.43 &1.04 &1.17 &0.78 &0.92\\
300 &1.08 &1.25 &0.94 &1.08 &0.75 &0.91
    &1.12 &1.20 &1.06 &1.21 &0.96 &0.94\\
\hline
\end{tabular}
\end{table}

\begin{table}[!ht]
\setlength{\abovecaptionskip}{-0.03pt}
\centering
\caption{\label{power1}Power (\%) comparison of the two LRTs,  $R_n$ and $R_n^*$, $KS$, and $AD$. The random samples are generated from model
\eqref{model}, in which $f_1=N(0,1)$ and $f_2=N(0.5,1)$ in Case I; $f_1=N(0,1)$ and $f_2=N(0.5,1.25^2)$ in Case II; and $f_1=N(0.5,0.75)$ and $f_2=N(0.5,1.25)$ in Case III. The significance level is $\alpha=5\%$.}
\begin{tabular}{cc cccc cccc c}
\hline\hiderowcolors
Case I &$\theta=0.5$  &\multicolumn{4}{c}{$n=100$}
       &\multicolumn{4}{c}{$n=200$} & 100KL\\
\cmidrule(lr){2-2} \cmidrule(lr){3-6} \cmidrule(lr){7-10}\cmidrule(lr){11-11}
&$d$ &$R_n$  &$R_n^*$ &$KS$ &$AD$ &$R_n$ &$R_n^*$ &$KS$ &$AD$&\\
&5   &53.9 &66.9	&50.4 &42.4	&85.8 &90.7	&79.6 &80.0&2.89\\
&20  &45.1 &56.8	&39.4 &37.4	&80.1 &88.1	&73.5 &71.6&2.53\\
\cmidrule(lr){2-11}
&$\theta=0.7$  &\multicolumn{4}{c}{$n=100$}
&\multicolumn{4}{c}{$n=200$} & 100KL\\
\cmidrule(lr){2-2} \cmidrule(lr){3-6}  \cmidrule(lr){7-10}\cmidrule(lr){11-11}
&$d$ &$R_n$  &$R_n^*$ &$KS$ &$AD$ &$R_n$ &$R_n^*$ &$KS$ &$AD$\\
&5   &58.7 &69.6    &48.9 &47.2 &87.2 &92.4 &79.3 &81.3&2.91\\
&20  &46.2 &58.1	&40.1 &42.2	&81.6 &89.0	&74.2 &74.4&2.61\\
\hline
Case II &$\theta=0.5$  &\multicolumn{4}{c}{$n=100$}
&\multicolumn{4}{c}{$n=200$} & 100KL\\
\cmidrule(lr){2-2} \cmidrule(lr){3-6}  \cmidrule(lr){7-10}\cmidrule(lr){11-11}
&$d$ &$R_n$  &$R_n^*$ &$KS$ &$AD$ &$R_n$ &$R_n^*$ &$KS$ &$AD$\\
&5   &62.0 &55.7	&43.9 &38.5	&92.4 &85.0	&74.5 &73.2&3.45\\
&20  &54.1 &47.6	&35.1 &33.3	&87.3 &79.3	&64.3 &67.1&3.05\\
\cmidrule(lr){2-11}
&$\theta=0.7$  &\multicolumn{4}{c}{$n=100$}
&\multicolumn{4}{c}{$n=200$}& 100KL \\
\cmidrule(lr){2-2} \cmidrule(lr){3-6}  \cmidrule(lr){7-10}\cmidrule(lr){11-11}
&$d$ &$R_n$  &$R_n^*$ &$KS$ &$AD$ &$R_n$ &$R_n^*$ &$KS$ &$AD$\\
&5   &61.6 &54.0	&44.5 &40.1	&92.7 &85.2	&73.7 &73.8&3.48\\
&20  &55.2 &50.4	&34.8 &36.8	&89.2 &80.4	&65.7 &69.8&3.13\\
\hline
Case III &$\theta=0.5$  &\multicolumn{4}{c}{$n=100$}
&\multicolumn{4}{c}{$n=200$} & 100KL\\
\cmidrule(lr){2-2} \cmidrule(lr){3-6}  \cmidrule(lr){7-10}\cmidrule(lr){11-11}
&$d$ &$R_n$  &$R_n^*$ &$KS$ &$AD$ &$R_n$ &$R_n^*$ &$KS$ &$AD$\\
&5   &30.3 &4.40	&5.70 &8.00	&59.9 &4.20	&9.90 &11.4&1.54\\
&20  &23.2 &4.60	&6.00 &6.90	&51.8 &5.40	&7.10 &9.00&1.36\\
\cmidrule(lr){2-11}
&$\theta=0.7$  &\multicolumn{4}{c}{$n=100$}
&\multicolumn{4}{c}{$n=200$} & 100KL\\
\cmidrule(lr){2-2} \cmidrule(lr){3-6}  \cmidrule(lr){7-10}\cmidrule(lr){11-11}
&$d$ &$R_n$  &$R_n^*$ &$KS$ &$AD$ &$R_n$ &$R_n^*$ &$KS$ &$AD$\\
&5   &29.4 &5.70	&7.00 &8.10	&59.2 &4.20	&8.60 &12.1&1.55\\
&20  &26.1 &4.10	&6.10 &5.50	&49.5 &4.60	&9.30 &9.90&1.40\\
\hline
\end{tabular}
\end{table}

\begin{table}[!ht]
\setlength{\abovecaptionskip}{-0.03pt}
\centering
\caption{\label{power2}Power (\%) comparison of the two LRTs, $R_n$ and $R_n^*$, $KS$, and $AD$. The random samples are generated from model
\eqref{model}, in which $f_1=Logistic(0,1)$ and $f_2=Logistic(1,1)$ in Case IV; $f_1=Logistic(0,1)$ and $f_2=Logistic(0.8,1.35)$ in Case V; and $f_1=Logistic(0.5,1)$ and $f_2=Logistic(0.5,1.5)$ in Case VI. The significance level is $\alpha=5\%$.}
\begin{tabular}{ccccccccccc}
\hline\hiderowcolors
Case IV &$\theta=0.5$  &\multicolumn{4}{c}{$n=100$}
&\multicolumn{4}{c}{$n=200$}&100KL \\
\cmidrule(lr){2-2} \cmidrule(lr){3-6} \cmidrule(lr){7-10}\cmidrule(lr){11-11}
~ &$d$ &$R_n$  &$R_n^*$ &$KS$ &$AD$ &$R_n$ &$R_n^*$ &$KS$ &$AD$\\
~ &5   &69.4 &78.2	&67.1 &60.9	&95.8 &97.2	&92.5 &91.8&3.83\\
~ &20  &61.0 &69.5	&55.9 &53.9	&88.4 &94.6	&84.8 &86.7&3.36\\
\cmidrule(lr){2-11}
&$\theta=0.7$  &\multicolumn{4}{c}{$n=100$}
&\multicolumn{4}{c}{$n=200$}&100KL  \\
\cmidrule(lr){2-2} \cmidrule(lr){3-6} \cmidrule(lr){7-10}\cmidrule(lr){11-11}
~ &$d$ &$R_n$  &$R_n^*$ &$KS$ &$AD$ &$R_n$ &$R_n^*$ &$KS$ &$AD$\\
~ &5   &68.9 &77.5	&64.0 &60.5	&95.1 &98.0	&93.9 &91.9&3.86\\
~ &20  &59.0 &68.9	&58.2 &53.7	&89.3 &93.8	&88.5 &87.8&3.46\\
\hline
Case V &$\theta=0.5$  &\multicolumn{4}{c}{$n=100$}
&\multicolumn{4}{c}{$n=200$}&100KL  \\
\cmidrule(lr){2-2} \cmidrule(lr){3-6} \cmidrule(lr){7-10}\cmidrule(lr){11-11}
~ &$d$ &$R_n$  &$R_n^*$ &$KS$ &$AD$ &$R_n$ &$R_n^*$ &$KS$ &$AD$\\
~ &5   &62.1 &47.6	&36.7 &36.8	&91.1 &77.8	&69.1 &70.3&3.33\\
~ &20  &55.4 &42.2	&29.8 &35.5	&85.5 &66.8	&56.5 &64.9&2.94\\
\cmidrule(lr){2-11}
&$\theta=0.7$  &\multicolumn{4}{c}{$n=100$}
&\multicolumn{4}{c}{$n=200$} &100KL \\
\cmidrule(lr){2-2} \cmidrule(lr){3-6} \cmidrule(lr){7-10}\cmidrule(lr){11-11}
~ &$d$ &$R_n$  &$R_n^*$ &$KS$ &$AD$ &$R_n$ &$R_n^*$ &$KS$ &$AD$\\
~ &5   &62.5 &45.6	&39.4 &37.0	&93.3 &77.4	&69.5 &72.3&3.36\\
~ &20  &54.1 &37.8	&33.2 &30.2	&86.5 &66.7	&60.7 &64.2&3.02\\
\hline
Case VI &$\theta=0.5$  &\multicolumn{4}{c}{$n=100$}
&\multicolumn{4}{c}{$n=200$} &100KL \\
\cmidrule(lr){2-2} \cmidrule(lr){3-6} \cmidrule(lr){7-10}\cmidrule(lr){11-11}
~ &$d$ &$R_n$  &$R_n^*$ &$KS$ &$AD$ &$R_n$ &$R_n^*$ &$KS$ &$AD$\\
~ &5   &52.9 &5.20	&8.10 &10.3	&85.6 &5.00	&16.6 &23.6&2.78\\
~ &20  &48.2 &6.10	&8.50 &11.3	&77.8 &4.90	&11.2 &19.9&2.45\\
\cmidrule(lr){2-11}
&$\theta=0.7$  &\multicolumn{4}{c}{$n=100$}
&\multicolumn{4}{c}{$n=200$} &100KL \\
\cmidrule(lr){2-2} \cmidrule(lr){3-6} \cmidrule(lr){7-10}\cmidrule(lr){11-11}
~ &$d$ &$R_n$  &$R_n^*$ &$KS$ &$AD$ &$R_n$ &$R_n^*$ &$KS$ &$AD$\\
~ &5   &52.0 &5.50	&9.50 &10.9	&84.1 &7.20	&16.6 &24.1&2.80\\
~ &20  &44.1 &4.70	&7.70 &11.4	&79.4 &5.70	&10.5 &22.3&2.52\\
\hline
\end{tabular}
\end{table}

\begin{table}[ht]
\setlength{\abovecaptionskip}{-0.03pt}
\centering\footnotesize
\tabcolsep=0.8mm
\caption{\label{ks.pvalue} Kolmogorov--Smirnov test for the normality of
the first sample and the fourth sample for each of the 17 intervals.}
\begin{tabular}{ccccccc}
\hline\hiderowcolors
Interval  &RG472--RG246 &RG246--K5 &K5--U10 &U10--RG532 &RG532--W1 &W1--RG173 \\
1st sample &0.6909	&0.7656	&0.5038	&0.4486	&0.9314	&0.8568\\
4th sample &0.7311	&0.7079	&0.7211	&0.612	&0.3305	&0.4123\\
\hline
Interval  &RG173--RZ276	 &RZ276--Amy1B &Amy1B--RG146	
&RG146--RG345 &RG345--RG381	&RG381--RZ19 \\
1st sample &0.7896 &0.9132 &0.9118 &0.9483 &0.9598 &0.8445\\
4th sample &0.3951 &0.5781 &0.5247 &0.8668 &0.9705 &0.9821\\
\hline
Interval  & RZ19--RG690	&RG690--RZ730  &RZ730--RZ801 &RZ801--RG810	&RG810--RG331&\\
1st sample &0.774 &0.6008 &0.996	&0.9736	&0.9632\\
4th sample &0.9328 &0.6159 &0.9753 &0.6082 &0.3977\\
\hline
\end{tabular}
\end{table}

\begin{table}[ht]
\setlength{\abovecaptionskip}{-0.03pt}
\centering\footnotesize
\tabcolsep=0.8mm
\caption{\label{real-data}  $P$-values of $R_n$, $R_n^*$, $KS$, and $AD$ in the 17 intervals.
Here ``$R_n$--Normal" and  ``$R_n^*$--Normal" are the two LRTs under the normal kernel, and ``$R_n$--Logistic" and ``$R_n^*$--Logistic" are the two LRTs under the logistic kernel.
}
\begin{tabular}{ccccccc}
\hline\hiderowcolors
Interval  &RG472--RG246 &RG246--K5 &K5--U10 &U10--RG532 &RG532--W1 &W1--RG173 \\
$R_n$--Normal   &0.683 &0.265 &0.283 &0.209 &0.292 &0.232\\
$R_n^*$--Normal &0.995 &0.943 &0.641 &0.830 &0.906 &0.757\\
$R_n$--Logistic   &0.612 &0.324 &0.311 &0.292 &0.364 &0.370\\
$R_n^*$--Logistic &0.965 &0.945 &0.472 &0.703 &0.785 &0.722\\
$KS$     &0.285 &0.567 &0.290 &0.141 &0.538	&0.520\\
$AD$     &0.369 &0.475 &0.358 &0.083 &0.457 &0.289\\
\hline
Interval  &RG173--RZ276	 &RZ276--Amy1B &Amy1B--RG146	
&RG146--RG345 &RG345--RG381	&RG381--RZ19 \\
$R_n$--Normal      &0.088 &0.835 &0.711 &0.820 &0.573 &0.161\\
$R_n^*$--Normal   &0.684 &0.876 &0.716 &0.697 &0.304 &0.056\\
$R_n$--Logistic     &0.196 &0.823 &0.760 &0.871 &0.650 &0.095\\
$R_n^*$--Logistic   &0.786 &0.599 &0.651 &0.747 &0.354 &0.032\\
$KS$      &0.485 &0.678	&0.521 &0.800 &0.740	 &{0.089}\\
$AD$      &0.116 &0.686 &0.524 &0.631 &0.755 &{0.117}\\
\hline
Interval  & RZ19--RG690	&RG690--RZ730  &RZ730--RZ801 &RZ801--RG810	&RG810--RG331&\\ \hline
$R_n$--Normal    &0 &0 &0 &0 &0 & \\
$R_n^*$--Normal   &0 &0 &0 &0 &0 & \\
$R_n$--Logistic      &0 &0 &0 &0 &0 & \\
$R_n^*$--Logistic  &0 &0 &0 &0 &0 & \\
$KS$      &0.001 &0 &0 &0 &0 &\\
$AD$      &0 &0 &0 &0 &0 &\\
\hline
\end{tabular}
\end{table}

\end{document}